\documentclass[journal,onecolumn]{IEEEtran}

\usepackage[utf8]{inputenc} 
\usepackage[T1]{fontenc}    
\usepackage{hyperref}       
\usepackage{url}            
\usepackage{booktabs}       
\usepackage{wrapfig}		
\usepackage{overpic}
\usepackage{amsfonts}       
\usepackage{nicefrac}       
\usepackage{float}			
\usepackage{microtype}      
\usepackage{multirow}
\usepackage{amsmath, amssymb, amsthm}
\usepackage[ruled,vlined,linesnumbered]{algorithm2e}
\usepackage{comment}
\usepackage{caption}
\usepackage{subcaption}
\usepackage{graphicx}
\usepackage{color}
\usepackage{cite}
\usepackage{bbm}
\usepackage{enumitem,kantlipsum}

\newcommand{\R}{\widetilde{R}}

\newcommand{\V}{\mathcal{V}}
\newcommand{\VG}{\mathcal{V}^{\Gamma}}
\newcommand{\MM}{\mathbf{M}}
\newcommand{\W}{\mathcal{W}}
\newcommand{\A}{\mathbf{A}}
\newcommand{\E}{\mathbb{E}}
\newcommand{\AH}{\mathbf{L}}
\newcommand{\AHGk}{\mathbf{L}^{(k)}_{\Gamma}}
\newcommand{\AHG}{\mathbf{L}_{\Gamma}}

\newcommand{\Qm}{Q_{\max}}

\newcommand{\HH}{\mathbf{H}}

\newcommand{\B}{\mathcal{B}}

\newcommand{\bb}{\mathbf{b}}
\newcommand{\Z}{\mathbf{Z}}

\newcommand{\z}{\mathbf{z}}

\newcommand{\DD}{\mathbb{D}}

\newcommand{\PP}{\mathbb{P}}

\newcommand{\p}{\mathbf{p}}

\newcommand{\m}{\mathbf{m}}
\newcommand{\M}{\mathcal{M}}
\newcommand{\EE}{\mathcal{E}}
\newcommand{\TT}{\mathcal{T}}
\newcommand{\T}{\mathbf{T}}

\newcommand{\hZv}{\widehat{Z}_{v}}

\newcommand{\hZZ}{\widehat{\mathbf{Z}}^{(0)}_{\sim v}}
\newcommand{\hzz}{\widehat{\mathbf{z}}^{(0)}_{\sim v}}
\newcommand{\zsv}{\mathbf{z}_{\sim v}}

\newcommand{\Ii}{\mathcal{I}^{\mathrm{in}}}
\newcommand{\Io}{\mathcal{I}^{\mathrm{out}}}
\newcommand{\U}{\mathcal{U}}

\newcommand{\G}{\mathcal{G}}

\DeclareMathOperator*{\argmax}{arg\,max}
\DeclareMathOperator*{\argmin}{arg\,min}

\newcommand{\mc}[1]{{\color{black}#1}}

\newtheorem{theorem}{Theorem}
\newtheorem{lemma}{Lemma}

\newtheorem{remark}{Remark}
\newtheorem{definition}{Definition}
\newtheorem{example}{Example}

\ifCLASSINFOpdf

\else

\fi

\begin{document}
\title{Exact Recovery in the General Hypergraph Stochastic Block Model}

\author{
        Qiaosheng Zhang,
        Vincent~Y.~F.~Tan,~\IEEEmembership{Senior~Member,~IEEE}
\thanks{Qiaosheng Zhang is with Shanghai Artificial Intelligence Laboratory, Shanghai 200232, China (e-mail: zhangqiaosheng@pjlab.org.cn).} 
\thanks{Vincent~Y.~F.~Tan is with the Department of Mathematics, National University
of Singapore, Singapore 119076, and the Department of Electrical and
Computer Engineering, National University of Singapore, Singapore 117583,
Singapore (e-mail: vtan@nus.edu.sg).}
}


\maketitle

\IEEEpeerreviewmaketitle

\begin{abstract}
This paper investigates fundamental limits of exact recovery in the general $d$-uniform hypergraph stochastic block model ($d$-HSBM), wherein $n$ nodes are partitioned into $k$ disjoint communities with relative sizes $(p_1,\ldots, p_k)$. Each subset of nodes with cardinality $d$ is generated independently as an order-$d$  hyperedge with a certain probability that depends on the ground-truth communities that the $d$ nodes belong to. The goal is to exactly recover the $k$ hidden communities based on the observed hypergraph. We show that there exists a sharp threshold  such that exact recovery is achievable above the threshold and impossible below the threshold (apart from a small regime of parameters that will be specified  precisely).  This threshold is represented in terms of a quantity which we term as the generalized Chernoff-Hellinger divergence between communities. Our result for this general model recovers prior results for the standard SBM and $d$-HSBM with two symmetric communities as special cases. En route to proving our achievability results, we develop a  polynomial-time two-stage algorithm  that meets the threshold. The first stage adopts a certain hypergraph spectral clustering method to obtain a coarse estimate of communities, and the second stage refines each node individually via local refinement steps to ensure exact recovery. 

\end{abstract}


\flushbottom


\section{Introduction} 
 
The stochastic block model (SBM)~\cite{holland1983stochastic} is a celebrated random graph model that has been widely studied for the community detection problem, and the objective therein is to partition $n$ nodes into $k$ disjoint communities (a.k.a., clusters) based on the randomly generated graph. The recent \emph{award-winning}\footnote{The Information Theory Society Paper Award, 2020.} papers~\cite{abbe2015exact,mossel2015consistency} discovered a phase transition phenomenon for \emph{exact recovery} (i.e., all the nodes are required to be classified correctly) in the SBM with two symmetric communities. That is, there is a sharp threshold such that exact recovery is achievable above the threshold, and  impossible below the threshold. This  phase transition phenomenon was later extended to the \emph{general SBM} with $k\ge2$ communities and without imposing symmetric structures~\cite{abbe2015community, abbe2015recovering}. Popularized by these breakthroughs, community detection in the SBM and its variants have then received significant attention, and many works have  progressively contributed to this field by considering the information-theoretic limits of some variants of the SBM~\cite{hajek2017information,reeves2019geometry,jog2015information,yun2016optimal}, efficient algorithms with theoretical guarantees~\cite{yun2014accurate, chin2015stochastic, hajek2016achieving, hajek2016achieving2,montanari2016,caltagirone2017recovering, agarwal2017multisection,perry2017semidefinite}, the effect of side information~\cite{asadi2017compressing,saad2018community,saad2020recovering,saad2018exact,mayya2019mutual}, etc. We refer the readers to~\cite{abbe2017community} for a comprehensive survey.

While most prior works focused on community detection on graphs, it is also of keen interest to study community detection on \emph{hypergraphs}. This is because  \emph{higher-order} relational information among multiple nodes, which can naturally be captured by hypergraphs, is ubiquitous in many applications. For example, friendships between users in social networks can be captured by graphs, but chat groups are usually represented by hyperedges in hypergraphs. Authors in co-authorship networks can also be connected by hyperedges. There are also applications in computer vision (such as object recognition and image registration) that are concerned  with the \emph{point-set matching} problem~\cite{chertok2010efficient,liu2010robust}, which aims to find strongly connected components in a uniform hypergraph. Motivated by these applications, in recent years some efforts have been expended to advance our understanding of community detection on hypergraphs. In particular, Ghoshdastidar and Dukkipati~\cite{ghoshdastidar2014consistency} first proposed a random hypergraph model called the \emph{$d$-uniform hypergraph stochastic block model} ($d$-HSBM), in which each subset of nodes with cardinality $d$ is generated independently as an \emph{order-$d$ hyperedge} with a certain probability that depends on the communities that the $d$ nodes belong to. Subsequent researchers further investigated the recovery limits of  $d$-HSBMs by developing various hypergraph clustering algorithms (such as spectral clustering methods~\cite{ghoshdastidar2015provable,ghoshdastidar2015spectral,ghoshdastidar2017consistency,pal2019community,cole2020exact,ahn2017information,ahn2018hypergraph}, semidefinite programming-based methods~\cite{kim2018stochastic,lee2020robust,kim2017community}, tensor decomposition-based methods~\cite{ke2019community}, approximate-message passing algorithms~\cite{angelini2015spectral,lesieur2017statistical}, etc) with theoretical guarantees, and characterizing the minimax misclassification proportion~\cite{chien2019minimax,lin2017fundamental,chien2018community} as well as the exact recovery criterion for the special case of two symmetric communities~\cite{kim2018stochastic}.\footnote{While this work mainly focuses on the $d$-HSBM, we are also aware that the clustering problem has also been explored in other hypergraph models, such as the non-uniform HSBM~\cite{ghoshdastidar2015provable,ghoshdastidar2015spectral,ghoshdastidar2017consistency}, the generalized censored block model~\cite{ahn2019community}, the sub-hypergraph models~\cite{liang2021information,yuan2021information}, etc.}

Although an information-theoretic limit for exact recovery has been derived in~\cite{kim2018stochastic}, their setting considering only two equal-sized communities and symmetric hyperedge generation probabilities is rather restrictive for real-world applications, and  a general theory for exact recovery in the $d$-HSBM is still lacking. Motivated by this gap in the literature, in this work we consider the exact recovery criterion in the \emph{general} $d$-HSBM.  Our problem setting and the distinctions compared to other works are summarized as follows.
\begin{enumerate}
	\item Nodes  are partitioned into $k \ge 2$ non-overlapping communities, in which each node is assigned to one of these $k$ communities with probabilities $\{p_i\}_{i=1}^k$. This generalizes the setting in~\cite{kim2018stochastic,pal2019community} in which $k = 2$ and $p_1 = p_2$, and the setting in~\cite{cole2020exact, chien2019minimax,lin2017fundamental,chien2018community} wherein the $k$ communities are of equal or approximately equal sizes (i.e., $p_i \approx p_j$ for all $1\le i<j \le k$).    
	
	\item The probability that an order-$d$ hyperedge appears depends on the number of nodes in each of the $k$ communities (which is quantified by a length-$k$ vector $(T_1, T_2, \ldots, T_k)$ with $\sum_{i} T_i = d$, where $T_i$ is the number of nodes in community $i$). In contrast, many prior works with theoretical guarantees consider more restrictive assumptions on the probability that a hyperedge is present. For example, in~\cite{ahn2018hypergraph,pal2019community,cole2020exact,kim2018stochastic,lee2020robust}  the hyperedge probability can only take \emph{two} values depending on whether all the $d$ nodes belong to the same community. Although other works such as~\cite{chien2019minimax,lin2017fundamental,chien2018community} relax the restrictions in~\cite{ahn2018hypergraph,pal2019community,cole2020exact,kim2018stochastic,lee2020robust}, they are nevertheless particularizations of our general HSBM model. We are aware that the assumption on hyperedge probabilities made in~\cite{ghoshdastidar2014consistency,ghoshdastidar2015provable,ghoshdastidar2015spectral,ghoshdastidar2017consistency} is similar to ours, but their focus is to characterize the performance of a hypergraph spectral clustering method contained therein in terms of the fraction of misclassified nodes, whereas we aim to quantify necessary and sufficient conditions for {\em exact recovery} (see Definition~\ref{def:exact} for details).
	
	\item Based on the observed $d$-uniform hypergraph, the learner is tasked to achieve  exact recovery of the hidden partition, i.e., all the $n$  nodes should be assigned to the ground-truth communities that they belong to with high probability as the size of the hypergraph grows.  We are also interested in deriving an algorithm-agnostic impossibility result that matches the performance guarantee of the learner's algorithm.  
	
\end{enumerate}

\subsection{Main Contributions}

The main contributions and key technical challenges of this work are summarized as follows.
\begin{enumerate}
	\item We establish a phase transition for the general $d$-HSBM (see Theorems~\ref{thm:converse} and~\ref{thm:achievability}), apart from a small subset of $d$-HSBMs that contains communities whose so-called \emph{second-order degree profiles} are identical (to be specified in Section~\ref{sec:distance}). That is, there is a sharp threshold such that exact recovery is possible above the threshold, and impossible below the threshold. This threshold is represented in terms of a quantity which we term as the \emph{generalized Chernoff-Hellinger (GCH) divergence} between different communities, and is a generalization of the CH-divergence discovered in~\cite{abbe2015community} for the SBM. Our result also recovers the exact recovery criterions for the SBM~\cite{abbe2015exact,mossel2015consistency} and $d$-HSBM with two symmetric communities~\cite{kim2018stochastic} as special cases. The techniques for proving the fundamental limits are  inspired by~\cite{abbe2015community}; however, dealing with hypergraphs requires us to carefully characterize different types of hyperedges that are induced by complicated community relations.
	
	\item We develop a polynomial-time algorithm that meets the information-theoretic limits. This implies that there is no information-computation gap for exact recovery in the general $d$-HSBM (apart from the aforementioned small regime of parameters). Our two-stage algorithm consists of a \emph{hypergraph spectral clustering} step in the first stage to ensure \emph{almost exact recovery}\footnote{In the literature, ``almost exact recovery'' is sometimes also called ``weak consistency'', and ``exact recovery'' is called ``strong consistency''.} (see Definition~\ref{def:almost}). It then performs local refinement steps for each of the $n$ nodes in the second stage to ensure exact recovery. To circumvent the problem that conditioned on the success of the first stage certain \emph{a priori} independent random variables become dependent, we adopt a \emph{hypergraph splitting} technique to split the hypergraph into two sub-hypergraphs (see Section~\ref{sec:split}), such that the two stages can be run on the two independent sub-hypergraphs respectively, preserving the independence of the two stages to facilitate the analysis. Although this technique is not new, our analytical method is different from previous analyses (such as~\cite{abbe2015community}). We prove that with high probability over splitting of the given hypergraph into two sub-hypergraphs, desirable properties of the resultant sub-hypergraphs are preserved, which further guarantees the success of the two stages. This new analytical method for analyzing multi-stage algorithms may be of independent interest. Algorithm~\ref{algorithm:1} can also be improved to an \emph{agnostic algorithm} that does not require the knowledge of model parameters (see Remark~\ref{remark:agnostic}).
	
	\item A main technical challenge lies in the development and analysis of an efficient algorithm that leads to almost exact recovery (for the first stage) for the general $d$-HSBM. To the best of our knowledge, such an  algorithm with accompanying guarantees is lacking in the literature. Thus, the hypergraph spectral clustering method developed here and its analysis may be of independent interest. We are aware that various clustering algorithms have been developed. However, theoretical guarantees (on the fraction of  misclassified nodes) are usually restricted to special classes of HSBMs and they do not readily apply to the general $d$-HSBM. For example, the performance of the spectral clustering method in~\cite{chien2019minimax} depends on the $k$-th largest singular value of a specific matrix, but bounding this value turns out to be non-trivial for general $d$-HSBMs. The semidefinite programming-based method in~\cite{kim2018stochastic} is only applicable to symmetric settings. Our new algorithm overcomes these stumbling blocks by leveraging and  judiciously combining  various ideas from prior works for the SBM~\cite{yun2016optimal, yun2014accurate} and the HSBM~\cite{chien2019minimax}. Our theoretical result (Theorem~\ref{thm:weak}) shows that, with probability approaching one, all but a vanishing fraction of the $n$ nodes can be assigned to their true communities (i.e., almost exact recovery is achieved) in the general $d$-HSBM.        
\end{enumerate}


\subsection{Organization}
We describe the general $d$-HSBM and the exact recovery criterion in Section~\ref{sec:model}, and provide our main results (with accompanying discussions) in Section~\ref{sec:main}. Our computationally efficient two-stage algorithm is introduced in Section~\ref{sec:algorithm}, and its theoretical guarantee is formally established in Section~\ref{sec:proof}. The converse part is proved in Section~\ref{sec:converse}. Section~\ref{sec:conclusion} concludes this work and proposes several directions that are fertile avenues for future research.

\section{Preliminaries and Problem Statement} \label{sec:model}

\subsection{Notation}
For any integer $m \ge 1$, let $[m]$ represent the set of integers $\{1,2,\ldots,m\}$, and $\mathcal{S}_m$ be the set of all permutations from $[m]$ to $[m]$.
Random variables and their realizations are respectively denoted by upper-case and lower-case letters, while vectors, matrices, and tensors  are denoted by boldface letters. For a length-$n$ vector $\mathbf{x}$, let $x_i$ denote its $i$-th element, and $\mathbf{x}_{\sim i}$ denote the length-$(n-1)$ sub-vector that excludes $x_i$.  For a matrix $\A$, its operator norm and Frobenius norm are respectively represented by $\Vert \A \Vert_{\mathrm{op}}$
and $\Vert \A \Vert_{\mathrm{F}}$, and its $v$-th column is denoted by $\A_v$.

\subsection{The $d$-uniform hypergraph stochastic block model ($d$-HSBM)} \label{sec:HSBM}
Let $n \in \mathbb{N}$ be the number of nodes, and $k \ge 2$ be the number of non-overlapping communities. Each node $v \in [n]$ belongs to one of the $k$ communities, and is associated with a latent random variable $Z_v$ on $[k]$ with prior distribution $\p = (p_1, p_2, \ldots, p_k)$, where $\sum_{i\in[k]}p_i =1$. That is, if node $v$ belongs to community $i$, then $Z_v = i$. The length-$n$ vector $\Z = (Z_1, Z_2, \ldots, Z_n)$ thus represents the \emph{ground-truth community vector} of the $n$ nodes.  Furthermore, we define $\V_i \triangleq \{v \in [n]: Z_v = i\} $ as the collection of nodes that belong to the community $i$ (for  $i \in [k]$).

Let $d \ge 2$ be the \emph{order} of the hyperedges (i.e., the number of nodes contained in each hyperedge),  and $\W$ be the set of all order-$d$ hyperedges on $[n]$ (where $|\W| = \binom{n}{d}$). It is assumed that the hypergraph considered in this work only contains order-$d$ hyperedges; thus it is referred to as a \emph{$d$-uniform hypergraph}. 
The generation process (underlying statistical model) of our random $d$-uniform hypergraph $G = ([n], \EE)$ is as follows.  For each $e \in \W$, the probability that it appears in the hypergraph $G$ (i.e., $e \in \EE$) depends on the number of nodes in each community $\{\V_i\}_{i=1}^k$.  Formally, let 
\begin{align}
\TT \triangleq \{\T \in \mathbb{N}^k: T_1 + T_2 + \cdots + T_k = d \}
\end{align}
be the collection of length-$k$ vectors such that each vector $\T \in \TT$ (with $T_i$  representing the number of nodes in  $\V_i$) represents a possible \emph{community assignment} of $d$ nodes, where ``community assignment'' is referred to as  the number of nodes contained in each community. The generation of the hyperedges in $G$ is fully characterized by a set of numbers $\left\{Q_{\T} \right\}_{\T \in \TT} \subset \mathbb{R}_+$. The probability of a hyperedge $e \in \W$ appearing is $Q_{ \T(e) }\frac{\log n}{n^{d-1}}$, where $ \T(e) \in \TT$ denotes the  community assignment of the $d$ nodes in hyperedge $e$. That is, $\mathbf{T}(e)$ is the length-$k$ vector whose $i$-th entry represents the number of nodes in the hyperedge $e$ that belongs to the $i$-th community.

\begin{example}
Suppose $d = 3$, $k = 4$, $n = 8$, and $\V_1 = \{1,2\}$, $\V_2 = \{3,4\}$, $\V_3 = \{5,6\}$, $\V_4 = \{7,8\}$. We list three different order-$d$ hyperedges $e_1, e_2, e_3 \in \W$, as well as their community assignments in the table below. Although $e_1 \ne e_2$, the probabilities that $e_1 \in \EE$ and $e_2 \in \EE$ are the same since they have the same community assignment.  On the other hand, the probability that $e_3 \in \EE$  is in general different from that for $e_1$ and $e_2$.
\begin{center}
	\begin{tabular}{||c c c||} 
		\hline
		Hyperedges & Community assignments & Hyperedge probabilities  \\ [0.5ex] 
		\hline\hline
		$e_1 = (1,4,8)$ & $\T(e_1) = (1,1,0,1)$ & $Q_{(1,1,0,1)}(\log n)/n^{2}$  \\ 
		\hline
		$e_2 = (1,3,7)$ & $\T(e_2) = (1,1,0,1)$ & $Q_{(1,1,0,1)}(\log n)/n^{2}$  \\
		\hline
		$e_3 = (2,5,7)$ & $\T(e_3) = (1,0,1,1)$ & $Q_{(1,0,1,1)}(\log n)/n^{2}$  \\
		\hline
	\end{tabular} \label{table}
\end{center}
\end{example}

\vspace{10pt}

The reason why we consider the $\Theta(\frac{\log n}{n^{d-1}})$-regime for the connectivity probability is that it ensures the average degree of each node is $\Theta(\log n)$ and it was shown~\cite{ahn2018hypergraph,kim2018stochastic,chien2019minimax} that phase transition for exact recovery  occurs in this \emph{logarithmic average degrees regime}.  Furthermore, we define $Q_{\max} \triangleq \max_{\T \in \TT}Q_{\T}$ and $Q_{\min} \triangleq \min_{\T \in \TT}Q_{\T}$, and it is assumed that the parameters $Q_{\max},Q_{\min},k,d$ and $\{p_{i}\}_{i\in[k]}$ do not scale with $n$. We also note that several related works~\cite{ghoshdastidar2014consistency, ghoshdastidar2017consistency, ke2019community} allow the number of communities $k$ to diverge as $n$ grows.

Similar to the adjacency matrices for  graphs, any $d$-uniform hypergraph $G$ can be represented by an order-$d$  $n \times \cdots \times n$  \emph{adjacency tensor} $\A = [A_{\bb}]$, where $\bb = [b_1, \ldots, b_d] \in [n]^d$ is the access index of the element in the tensor. Here, $A_{\bb} \in \{0,1\}$, and $A_{\bb} = 1$ means the presence of the hyperedge corresponding to the $d$ nodes in $\bb$. In particular, $A_{\bb} = 0$ if the $d$ elements in $\bb$ are \emph{not} distinct (since each hypergraph must contain $d$ nodes), and $A_{\bb} = A_{\bb'}$ if there exists a permutation $\pi$ such that $(b'_{\pi(1)}, b'_{\pi(2)}, \ldots, b'_{\pi(d)}) = (b_1, b_2, \ldots, b_d)$.

\subsection{Objective}
Given the observation of the hypergraph $G$ (or the adjacency tensor $\A$), the learner aims to use an estimator $\phi = \phi(G)$ to recover the \emph{partition} of the $n$ nodes into $k$ communities. The output of the estimator $\phi$ is denoted by $\widehat{\Z} =(\widehat{Z}_1, \widehat{Z}_2, \ldots, \widehat{Z}_n)$. We measure the accuracy of $\widehat{\Z}$ in terms of the \emph{misclassification proportion} $l (\widehat{\Z},\Z)$, which is defined as     
\begin{align}
l (\widehat{\Z},\Z) \triangleq  \min_{\pi \in \mathcal{S}_k}\frac{1}{n} \sum_{v \in[n]} \mathbbm{1}\left\{\widehat{Z}_v \ne \pi(Z_v) \right\}. \label{eq:misclassification}
\end{align}

\begin{definition}[Exact recovery] \label{def:exact}
An estimator $\phi$ is said to achieve exact recovery if it ensures that with probability $1 - o(1)$, the misclassification proportion $l (\widehat{\Z},\Z) = 0$.
\end{definition}
\begin{definition}[Almost exact recovery]\label{def:almost}
An estimator $\phi$ is said to achieve almost exact recovery if it ensures that with probability $1 - o(1)$, the misclassification proportion $l (\widehat{\Z},\Z) \to 0$ as $n$ tends to infinity.
\end{definition}

\section{Main Results and Discussions} \label{sec:main}

We first introduce several notations that are useful for stating our main results. Let 
\begin{align}
\M \triangleq \{\m \in \mathbb{N}^k: m_1 + m_2 + \cdots + m_k = d-1\}
\end{align}
be the collection of length-$k$ vectors such that the sum of the $k$ elements equals $d-1$. Each element in $\M$ represents one possible community assignment of $d-1$ nodes. For each $\m \in \M$, we define 
\begin{align}
R_{\m} \triangleq \prod_{s=1}^k \binom{np_s}{m_s} \quad \text{and} \quad R'_{\m} \triangleq \frac{R_{\m}}{n^{d-1}}
\end{align}
as the expected number (and normalized expected number)  of combinations of $d-1$ nodes that have community assignment $\m$. Note that $R_{\m} = \Theta(n^{d-1})$ since $\binom{np_s}{m_s} = \Theta(n^{m_s})$ and $\sum_{s\in[k]}m_s = d-1$, and thus $R'_{\m} = \Theta(1)$. Suppose a node belongs to $\V_i$ (where $i \in [k]$) and other $d-1$ nodes have community assignment $\m = (m_1, \ldots, m_i, \ldots, m_k) \in \M$, then the joint community assignment of these $d$ nodes is denoted by 
\begin{align}
\m \oplus i \triangleq (m_1, \ldots, m_{i-1}, m_i+1, m_{i+1}, \ldots, m_k) \in \TT.
\end{align}  
For instance, when $d = 3$, $k = 4$, $\m = (1,1,0,0)$ and $i = 3$, we have $\m \oplus i = (1,1,1,0)$.

\subsection{Separation between communities} \label{sec:distance}
\subsubsection{Degree profile}
For each community $\V_i$ (where $i \in [k]$), we define $\mu_{\m \oplus i}\triangleq R'_{\m}Q_{\m \oplus i}$ for each $\m \in \M$. The interpretation of $\mu_{\m \oplus i}$ is as follows: for any node $v \in \V_i$, the expected number of hyperedges that contain $v$ and have community assignment $\m \oplus i$ is $R_{\m} \cdot \frac{Q_{\m \oplus i} \log n}{n^{d-1}}$, thus  $\mu_{\m \oplus i}$ is the corresponding normalized quantity which scales as $\Theta(1)$.  We then refer to the collection $\{\mu_{\m \oplus i}\}_{\m \in \M}$ as the  \emph{degree profile} of community $i$. Intuitively, two communities are easier to be separated  if the degree profiles of these two communities are further apart. The discrepancy between any two communities in the $d$-HSBM can be measured in terms of the \emph{generalized Chernoff-Hellinger divergence (GCH-divergence)} between their degree profiles, which generalizes the CH-divergence for the SBM that was first discovered by Abbe and Sandon~\cite[Eqn. (3)]{abbe2015community}. 

\begin{definition}[GCH-divergence] \label{def:CH}
For any $i,j \in [k]$ such that $i \ne j$, we define the GCH-divergence between $i$ and $j$ as 
\begin{align}
D_+(i,j) \triangleq \max_{t \in [0,1]} \sum_{\m \in \M} t\mu_{\m \oplus i} + (1-t)\mu_{\m \oplus j} - \mu_{\m \oplus i}^t \mu_{\m \oplus j}^{1-t}, \label{eq:ch}
\end{align}
where $D_+(i,j)$ is a function of $\p$ and $\{Q_{\T}\}_{\T \in \TT}$.
\end{definition}
Note that $D_+(i,j) = 0$ if and only if the degree profiles of communities $\V_i$ and $\V_j$ are exactly the same, in which case the two communities are statistically indistinguishable.

\subsubsection{Second-order degree profile} 
For each community $\V_i$ (where $i \in [k]$), we define its \emph{second-order degree profile} as $\big\{\sum_{\m \in \M:m_s \ge 1}m_s \mu_{\m \oplus i} \big\}_{s \in [k]}$, where each element $\sum_{\m \in \M:m_s \ge 1}m_s \mu_{\m \oplus i}$ represents the normalized expected number of hyperedges that contain two fixed nodes belonging to $\V_i$ and $\V_s$ respectively. When the second-order degree profile of two communities are exactly the same, our analysis also shows that there may be some inherent difficulties in distinguishing them. Formally, we define $\Xi$ as the subset of model parameters $\big(\p, \{Q_{\T}\}_{\T \in \TT}\big)$ such that there exist two communities having the same second-order degree profiles, i.e.,
\begin{align}
\Xi \triangleq \left\{\Big(\p, \{Q_{\T}\}_{\T \in \TT}\Big): \exists i \ne j \text{ such that } \sum_{\m \in \M: m_s \ge 1} m_s \mu_{\m \oplus i}  = \sum_{\m \in \M: m_s \ge 1} m_s \mu_{\m \oplus j} \text{ for all } s\in[k] \right\}.
\end{align}

\subsection{Main results and discussions} \label{sec:result}

\begin{theorem}[Converse] \label{thm:converse}
It is impossible to achieve exact recovery when the model parameters $\big(\p, \{Q_{\T}\}_{\T \in \TT}\big)$ satisfy 
\begin{align}
\min_{i,j\in [k]: i \ne j} D_+(i,j) < 1.
\end{align}
\end{theorem}

\begin{theorem}[Achievability] \label{thm:achievability}
Assume that the model parameters $\big(\p, \{Q_{\T}\}_{\T \in \TT}\big) \notin \Xi$. Then the polynomial-time two-stage algorithm (Algorithm~\ref{algorithm:1}) achieves exact recovery when the model parameters $\big(\p, \{Q_{\T}\}_{\T \in \TT}\big)$ satisfy 
	\begin{align}
	\min_{i,j\in [k]: i \ne j} D_+(i,j) > 1.
	\end{align}
\end{theorem}
Some remarks on Theorems~\ref{thm:converse} and~\ref{thm:achievability} are in order.

\begin{enumerate}
\item 
For community detection in the $d$-HSBM, most of the settings considered in prior works, such as the one that $\{Q_{\T}\}_{\T \in \TT}$ can only take two values depending on whether $d$ nodes belong to the same community~\cite{pal2019community,cole2020exact,ahn2017information,ahn2018hypergraph,kim2018stochastic,lee2020robust}, satisfy $\big(\p, \{Q_{\T}\}_{\T \in \TT}\big) \notin \Xi$. Thus, our result is a strict generalization of these existing works. An example for the case that $\big(\p, \{Q_{\T}\}_{\T \in \TT}\big) \in \Xi$ is as follows. Suppose $k = 2$, $d = 3$, $\p = (p_1,p_2)$, and $\{Q_{\T}\}_{\T \in \TT} = \{Q_{(3,0)},Q_{(2,1)}, Q_{(1,2)}, Q_{(0,3)} \}$. When $(\p,\{Q_{\T}\}_{\T \in \TT})$ satisfies $p_1 = p_2 = 1/2$, $Q_{(3,0)} = Q_{(1,2)}$ and $Q_{(2,1)} = Q_{(0,3)}$, one can check that the two communities $\V_1$ and $\V_2$ have the same second-order degree profile, thus $\big(\p, \{Q_{\T}\}_{\T \in \TT}\big) \in \Xi$.

\item  When $\big(\p, \{Q_{\T}\}_{\T \in \TT}\big) \in \Xi$ and $\min_{i,j\in [k]: i \ne j} D_+(i,j) > 1$, it remains open whether exact recovery is possible. However, we would like to point out that this scenario does not apply to the SBM (equivalently, the $2$-HSBM), since the condition $\min_{i,j\in [k]: i \ne j} D_+(i,j) > 1$ immediately implies that $\big(\p, \{Q_{\T}\}_{\T \in \TT}\big) \notin \Xi$ in this setting. Our sharp threshold is applicable to all model parameters when $d=2$.   


\item When $\big(\p, \{Q_{\T}\}_{\T \in \TT}\big) \notin \Xi$, the first stage of Algorithm~\ref{algorithm:1} ensures almost exact recovery via hypergraph spectral clustering (as shown in Theorem~\ref{thm:weak}), and the condition $\min_{i,j\in [k]: i \ne j} D_+(i,j) > 1$ is the criterion for Stage 2 (local refinement steps) to succeed.  Roughly speaking, performing a local refinement step for each node is equivalent to performing a hypothesis test with independent but non-identically distributed samples. The corresponding error probability can be represented by a variant of the \emph{Chernoff information}~\cite[Chapter 11.9]{cover2012elements}, and this further reduces to $n^{-\min_{i,j\in [k]: i \ne j} D_+(i,j)}$ which is in the form of the GCH-divergence. Thus, when $\min_{i,j\in [k]: i \ne j} D_+(i,j) > 1$, taking a union bound over the $n$ nodes results in a vanishing error probability (i.e., exact recovery is achieved). 

\item Our analysis of the first stage is not able to handle the case in which $\big(\p, \{Q_{\T}\}_{\T \in \TT}\big) \in \Xi$ because the key step in the hypergraph spectral clustering method is to map the order-$d$ adjacency tensor $\A$ to an $n \times n$ matrix $\AH$ (defined in~\eqref{eq:L} below), and the subsequent clustering algorithm critically relies on the discrepancy between the columns of $\AH$ (which corresponds to the second-order degree profiles of communities). We conjecture that this issue may be circumvented if one directly applies clustering algorithms on the adjacency tensor (such as the method proposed in~\cite{ke2019community}), and the exact recovery threshold $\min_{i,j\in [k]: i \ne j} D_+(i,j) =1$ holds without the assumption that the second-order degree profiles of any two communities are  distinct.

\item The algorithm performance also depends on the value of $d$. As $d$ increases, the computational complexity of constructing $\mathbf{L}_{\Gamma}$ increases accordingly. If the hyperedge probabilities $\{Q_{\mathbf{T}}\}_{\mathbf{T}\in \mathcal{T}}$ were unknown a priori, a larger value of $d$ would also increase the difficulty of learning $\{Q_{\mathbf{T}}\}_{\mathbf{T}\in \mathcal{T}}$, since $|\mathcal{T}|$ increases exponentially with $d$.  
\end{enumerate}

\subsection{Recovering prior results from Theorems~\ref{thm:converse} and~\ref{thm:achievability}}
To the best of our knowledge, the sharp threshold established by Theorems~\ref{thm:converse} and~\ref{thm:achievability} is the most general result for exact recovery in the SBM/HSBM literature. As discussed below, several problem settings investigated in prior works are subsumed by our result, and the thresholds derived in the  prior works can be recovered from Theorems~\ref{thm:converse} and~\ref{thm:achievability}.   

\subsubsection{Exact recovery in the SBM~\cite{abbe2015community}} The SBM considered in~\cite{abbe2015community} corresponds to our $d$-HSBM with $d = 2$. In~\cite{abbe2015community}, the prior distribution of each node is also $\p = (p_1, p_2, \ldots, p_k)$, and the edge probabilities are characterized by $\big\{Q_{i,j}\frac{\log n}{n} \big\}_{i,j \in [k]}$, where $Q_{i,j}$ corresponds to edges that contain nodes in $\V_i$ and $\V_j$. The authors of~\cite{abbe2015community} showed that the threshold for exact recovery is 
\begin{align}
	\min_{i,j\in[k]:i\ne j} \max_{t\in [0,1]} \sum_{s \in [k]} p_s\left[t Q_{s,i} + (1-t)Q_{s,j} - Q_{s,i}^tQ_{s,j}^{1-t} \right] = 1.  \label{eq:abbe}
\end{align}  
In our setting with $d = 2$, the set $\M = \{\m_1, \m_2, \ldots, \m_{k}\}$ contains $k$ distinct length-$k$ vectors, where each $\m_s \triangleq (0,\ldots,0,1,0,\ldots,0)$ contains a single one which is at the $s$-th location. By noting that $\mu_{\m \oplus i} = p_s Q_{s,i}$ (resp. $\mu_{\m \oplus j} = p_s Q_{s,j}$)   in~\cite{abbe2015community} and  $\big(\p, \{Q_{\T}\}_{\T \in \TT}\big)\notin \Xi$ when 	$\min_{i,j\in [k]: i \ne j} D_+(i,j) > 1$, we recover their threshold stated in Eqn.~\eqref{eq:abbe} from Theorems~\ref{thm:converse} and~\ref{thm:achievability}.

\begin{remark}
{\em One major distinction between the algorithms in~\cite{abbe2015community} and this work is the initialization step (Stage~1). We use the spectral clustering method while~\cite{abbe2015community} uses the so-called \emph{sphere comparison} algorithm. The main idea of the sphere comparison algorithm is to determine whether two nodes belong to a same community by counting the common neighbors at a large enough depth between them. While it works well for regular graphs, generalizing it to hypergraphs may be non-trivial.}
\end{remark}

\subsubsection{Exact recovery in the $d$-HSBM with two symmetric communities~\cite{kim2018stochastic}}
The model considered in~\cite{kim2018stochastic} is a special $d$-HSBM with two equal-sized communities that have symmetric structures. It corresponds to our general $d$-HSBM with $k = 2$, $p_1 = p_2 = 1/2$, and $\{Q_{\T}\}_{\T \in \TT} = \{Q_1,Q_2\}$ (where each hyperedge appears with probability $Q_1$ when $d$ nodes are in the same communities, and $Q_2$ otherwise). Kim, Bandeira, and Goemans~\cite{kim2018stochastic} showed  that the threshold for exact recovery is 
\begin{align}
	\frac{1}{2^{d-1}} \left(\sqrt{\frac{Q_1}{(d-1)!}} - \sqrt{\frac{Q_2}{(d-1)!}}  \right)^2 = 1.  \label{eq:hsbm}
\end{align}  
Specializing our result to this symmetric setting, we note that the set $\M = \{(d-1,0), (d-2,1), \ldots, (0,d-1) \}$ contains $d$ distinct length-$2$ vectors, and the values of $\mu_{\m\oplus 1}$ and $\mu_{\m\oplus 2}$ for all $\m \in \M$ can then be calculated. By noting that the model parameters $(\p, \{Q_{\T}\}_{\T \in \TT}) \notin \Xi$ and $t = 1/2$ maximizes the GCH-divergence in~\eqref{eq:ch} for symmetric SBMs and HSBMs, we recover the threshold stated in Eqn.~\eqref{eq:hsbm} from Theorems~\ref{thm:converse} and~\ref{thm:achievability}. Furthermore, we also recover the celebrated exact recovery threshold $|\sqrt{Q_1} - \sqrt{Q_2}| = \sqrt{2}$ for the SBM with two symmetric communities~\cite{abbe2015exact}, since it is a special case of~\cite{kim2018stochastic} for $d=2$.

\subsection{Comparisons with the results on the misclassification proportion in the HSBM~\cite{chien2019minimax}} \label{sec:chien}
The work~\cite{chien2019minimax} studied the \emph{fundamental limit of misclassification proportion} in $d$-HSBMs. Their model assumes that there are $k$ approximately equal-sized communities, and the hyperedge probabilities depend only on the sorted histogram vector (in descending order) of the community assignment vector (e.g., the community assignment vectors $\mathbf{T} = (d, 0, \ldots, 0)$ and $\mathbf{T}' = (0, \ldots, 0, d)$ correspond to a same sorted histogram vector $(d, 0, \ldots, 0)$). Thus, their model is a particularization of our general HSBM model. While the main focus of their work is to characterize the negative exponent of the misclassification proportion  $l(\widehat{\Z},\Z)$ (as defined in Eqn.~\eqref{eq:misclassification}), their results can also be applied to finding the exact recovery threshold by setting the negative exponent to be greater than $\log n$ (which means $l(\widehat{\Z},\Z) < 1/n$ and thus implies exact recovery). In the following, we show that the exact recovery thresholds derived in this work and~\cite[Theorem~3.1, Theorem~3.2]{chien2019minimax} are exactly the same when $d = 2$ and $d = 3$. When $d \ge 4$, the expressions in both works become highly complicated (and moreover,~\cite{chien2019minimax} did not provide the precise value of their expression for $d \ge 5$), thus it is difficult to make comparisons; however, we conjecture that the thresholds should still be the same for $d \ge 4$ due to the evidence shown for $d = 2$ and $d = 3$.

\subsubsection{Comparison for $d = 2$} For a valid comparison, we assume that (i) there are $k$ communities of equal sizes, and (ii) the hyperedge probabilities are either $q_1 = Q_1 \frac{\log n}{n}$ or $q_2 = Q_2 \frac{\log n}{n}$, depending on whether two nodes belong to a same community.
For the standard SBM with $d = 2$, the theoretical result in~\cite{chien2019minimax} reduces to the minimax misclassification proportion in~\cite{zhang2016minimax}, in which the negative exponent is dominated by $(n/k)\cdot I_{q_1q_2}$, where $I_{q_1q_2} \triangleq -2\log(\sqrt{q_1q_2}+\sqrt{1-q_1}\sqrt{1-q_2})$ and can further be simplified as
\begin{align}
-2\log(\sqrt{q_1q_2}+\sqrt{1-q_1}\sqrt{1-q_2}) &= -2\log\left\{ \sqrt{q_1q_2}+ \left(1-\frac{1}{2}q_1 + \mathcal{O}(q_1^2) \right) + \left(1-\frac{1}{2}q_2 + \mathcal{O}(q_2^2)\right) \right\} \notag \\
&= -2 \log \left\{1 - \left( \frac{1}{2}(q_1 + q_2) - \sqrt{q_1q_2} + \mathcal{O}(q_1q_2) \right) \right\} \notag \\
& = 2  \left( \frac{1}{2}(q_1 + q_2) - \sqrt{q_1q_2}\right) + \mathcal{O}(q_1q_2) \notag \\
&= (\sqrt{q_1} - \sqrt{q_2})^2  + \mathcal{O}(q_1q_2), \label{eq:mid}\\
&= (\sqrt{Q_1} - \sqrt{Q_2})^2 \frac{\log n}{n} + \mathcal{O}((\log n)^2 / n^2). \notag
\end{align}
The first equality follows from $\sqrt{1-x} = 1 - \frac{1}{2}x + \mathcal{O}(x^2)$ for $x \to 0$. For sufficiently large $n$, when the parameters satisfy 
\begin{align}
\frac{(\sqrt{Q_1} - \sqrt{Q_2})^2}{k} > 1,  \label{eq:d2}
\end{align}
the negative exponent of the misclassification proportion will be greater than $\log n$ (i.e., the misclassification proportion will be less than $1/n$), which implies exact recovery. 

On the other hand, when $d = 2$, the theoretical guarantee of exact recovery in our work reduces to the threshold in~\cite{abbe2015community}, which is exactly the condition given in~\eqref{eq:d2}. This means that for $d =2$, the exact recovery thresholds in~\cite{chien2019minimax} and this work are the same.

\subsubsection{Comparison for $d = 3$} For a valid comparison, we assume that (i) there are $k$ communities of equal sizes, and (ii) the hyperedge probabilities scale as $\Theta(\frac{\log n}{n^2})$, and depend only on the sorted histogram vector of the community assignment: 
\begin{itemize}
\item A hyperedge appears with probability $q_1 = Q_1 \frac{\log n}{n^2}$ if all three nodes belong to a same community;
\item A hyperedge appears with probability $q_2 = Q_2 \frac{\log n}{n^2}$ if only two nodes belong to a same community;
\item A hyperedge appears with probability $q_3 = Q_3 \frac{\log n}{n^2}$ if three nodes belong to three different communities.
\end{itemize}
We note that~\cite[Theorem 3.1]{chien2019minimax} guarantees that the misclassification proportion between the true and estimated labels is at most $\exp\{-(1-\xi_n) [\frac{n^2}{2k^2}I_{q_1q_2} + \frac{n^2 (k-2)}{k^2}I_{q_2q_3}]  \}$ with high probability, where $I_{q_i q_j} = -2\log(\sqrt{q_iq_j}+\sqrt{1-q_i}\sqrt{1-q_j})$ and $\xi_n \to 0$ as $n \to \infty$. To ensure exact recovery, the negative exponent should satisfy 
	\begin{align}
	\frac{n^2}{2k^2}I_{q_1q_2} + \frac{n^2 (k-2)}{k^2}I_{q_2q_3} > \log n. \label{eq:condition}
	\end{align}
Next, we figure out the condition under which~\eqref{eq:condition} holds. Recalling from~\eqref{eq:mid} that $I_{q_iq_j} = (\sqrt{q_i} - \sqrt{q_j})^2 +\mathcal{O}(q_iq_j)$, thus the LHS of~\eqref{eq:condition} can be expressed as 
\begin{align*}
&\frac{n^2}{2k^2} (\sqrt{q_1} - \sqrt{q_2})^2  + \frac{n^2 (k-2)}{k^2} (\sqrt{q_2} - \sqrt{q_3})^2 + \mathcal{O}(n^2 q_1 q_2) \\
&= \left( \frac{(\sqrt{Q_1} - \sqrt{Q_2})^2}{2k^2}  + \frac{(k-2)(\sqrt{Q_2} - \sqrt{Q_3})^2}{k^2} \right) \cdot  \log n + \mathcal{O}\left(\frac{(\log n)^2}{n^2} \right)
\end{align*}
For sufficiently large $n$, when the model parameters satisfy
\begin{align}
\frac{(\sqrt{Q_1} - \sqrt{Q_2})^2}{2k^2}  + \frac{(k-2)(\sqrt{Q_2} - \sqrt{Q_3})^2}{k^2} > 1, \label{eq:threshold}
\end{align} 
the misclassification proportion will be less than $1/n$, which implies exact recovery.

Next, we specialize our results to the setting of interest. Note that $\mathcal{M} = \big\{\mathbf{m} \in \mathbb{N}^k: m_1 + m_2 + \cdots + m_k = 2 \big\}$, and $R'_{\mathbf{m}} = 1/k^2$ if $\max_{l \in [k]} m_{l} = 1$, and $R'_{\mathbf{m}} = 1/(2k^2)$ if $\max_{l \in [k]} m_{l} = 2$. One can check that  the second-order degree profile condition is satisfied. Without loss of generality, we  focus on the first two communities: $\mathcal{V}_1$ with degree profile $\{\mu_{\mathbf{m} \oplus 1}\}_{\mathbf{m} \in \mathcal{M}}$ and $\mathcal{V}_2$ with degree profile $\{\mu_{\mathbf{m} \oplus 2}\}_{\mathbf{m} \in \mathcal{M}}$.
In the following, we consider $\mathbf{m} \in \mathcal{M}$ such that $\mu_{\mathbf{m} \oplus 1}$ and $\mu_{\mathbf{m} \oplus 2}$ are different:
\begin{itemize}
\item[--] When $\mathbf{m} = (2, 0, 0 \ldots, 0)$, we have $\mu_{\mathbf{m} \oplus 1} =Q_1/(2k^2)$ and $\mu_{\mathbf{m} \oplus 2} =Q_2/(2k^2)$;
\item[--] When $\mathbf{m} = (0, 2, 0 \ldots, 0)$, we have $\mu_{\mathbf{m} \oplus 1} =Q_2/(2k^2)$ and $\mu_{\mathbf{m} \oplus 2} =Q_1/(2k^2)$;
\item[--] When $\mathbf{m}$ satisfies $m_1 = 1, m_2 = 0$, and there exists only one index $l \in \{3,\ldots, k\}$ such that $m_l = 1$, we have $\mu_{\mathbf{m} \oplus 1} =Q_2/k^2$ and $\mu_{\mathbf{m} \oplus 2} =Q_3/k^2$;
\item[--] When $\mathbf{m}$ satisfies $m_1 = 0, m_2 = 1$, and there exists only one index $l \in \{3,\ldots, k\}$ such that $m_l = 1$, we have $\mu_{\mathbf{m} \oplus 1} =Q_3/k^2$ and $\mu_{\mathbf{m} \oplus 2} =Q_2/k^2$.
\end{itemize}
Thus, the GCH-Divergence $D_+(1,2)$ between the first two communities is 
\begin{align}
&\max_{t \in [0,1]} \left[t\frac{Q_1}{2k^2} + (1-t)\frac{Q_2}{2k^2} - \left( \frac{Q_1}{2k^2}\right)^t \left(\frac{Q_2}{2k^2} \right)^{1-t} \right] + \left[t\frac{Q_2}{2k^2} + (1-t)\frac{Q_1}{2k^2} - \left( \frac{Q_2}{2k^2}\right)^t \left(\frac{Q_1}{2k^2} \right)^{1-t} \right] \notag \\
&+(k-2)\left[t\frac{Q_2}{k^2} + (1-t)\frac{Q_3}{k^2} - \left( \frac{Q_2}{k^2}\right)^t \left(\frac{Q_3}{k^2} \right)^{1-t} \right] + (k-2)\left[t\frac{Q_3}{k^2} + (1-t)\frac{Q_2}{k^2} - \left( \frac{Q_3}{k^2}\right)^t \left(\frac{Q_2}{k^2} \right)^{1-t} \right], \notag
\end{align}
where the minimum is obtained at $t = 1/2$, yielding that $D_+(1,2) = \frac{(\sqrt{Q_1} - \sqrt{Q_2})^2}{2k^2}  + \frac{(k-2)(\sqrt{Q_2} - \sqrt{Q_3})^2}{k^2}$.
Finally, by symmetry one can show that $D_+(i,j) = D_+(1,2)$ for other pairs of $i,j \in [k]$. Therefore, the exact recovery threshold  is 
\begin{align}
\frac{(\sqrt{Q_1} - \sqrt{Q_2})^2}{2k^2}  + \frac{(k-2)(\sqrt{Q_2} - \sqrt{Q_3})^2}{k^2} > 1,
\end{align}
which is exactly the same as the threshold~\eqref{eq:threshold} derived in~\cite{chien2019minimax}.

\section{The two-stage algorithm for exact recovery} \label{sec:algorithm}

In this section, we present our  polynomial-time algorithm that is used to achieve the information-theoretic limit shown in Theorem~\ref{thm:achievability}. As mentioned in Section~\ref{sec:result}, our algorithm consists of two stages such that the first stage achieves almost exact recovery via the hypergraph spectral clustering method and the second stage achieves exact recovery via local refinement steps. This ``\emph{from global to local}'' principle has been employed in many contexts, such as community detection in the SBM~\cite{abbe2015community,yun2016optimal, yun2014accurate, chin2015stochastic, gao2017achieving} and HSBM~\cite{ahn2018hypergraph,kim2018stochastic,chien2019minimax,lin2017fundamental,chien2018community}, matrix completion~\cite{jain2013low,keshavan2010matrix,zhang2020mc2g}, etc. 
It is also worth noting that when analyzing two-stage algorithms, random variables that are initially independent may become dependent conditioned on the success of a preceding stage. To ameliorate this problem, we adopt the graph splitting technique (as described in Subsection~\ref{sec:split}) which is inspired by prior works on community detection~\cite{abbe2015community,chin2015stochastic,ahn2018hypergraph,vu2014simple}.   Our algorithm is described in detail in Algorithm~\ref{algorithm:1}.

\subsection{Graph Splitting} \label{sec:split}
Let $F = ([n], \W)$ be the \emph{complete $d$-uniform hypergraph} on node set $[n]$, and the hyperedge set $\W$ contains all the $\binom{n}{d}$ order-$d$ hyperedges (as defined in Section~\ref{sec:HSBM}). We randomly split $F$ into two sub-hypergraph $F_1 = ([n], \W_1)$ and $F_2 = ([n], \W_2)$. Each hyperedge in $\W$ is sampled to $\W_1$ with probability $\gamma_n/\log n$, and to $\W_2$ with probability $1 - (\gamma_n/\log n)$, where $\gamma_n$ can be any value in $\omega(1) \cap o(\log n)$. For concreteness we set $\gamma_n = \sqrt{\log n}$. Note that $\W_2$ is the complement of $\W_1$.  This splitting process is independent of the generation of the hypergraph  $G = ([n], \EE)$ (which is generated according to $\p$ and $\left\{Q_{\T} (\log n)/n^{d-1} \right\}_{\T \in \TT}$). We then define $G_1 = ([n], \EE_1)$ as the \emph{sub-HSBM} that is generated on the hyperedge set $\W_1$ of $F_1$, where  $\EE_1 = \EE \cap \W_1$ is the intersection of the hyperedge sets of the HSBM $G$ and the sub-hypergraph $F_1$. Similarly, we define $G_2 = ([n], \EE_2)$ as the sub-HSBM that is generated on the hyperedge set $\W_2$ of  $F_2$, where  $\EE_2 = \EE \cap \W_2$.

\begin{algorithm}
	\SetAlgoLined
	\SetKwInOut{Input}{Input}\SetKwInOut{Output}{Output}
	\SetKwData{Left}{left}\SetKwData{This}{this}\SetKwData{Up}{up}
	\textbf{Input:} Hypergraphs $G_1, G_2$, number of communities $k$, $\gamma_n = \sqrt{\log n}$, radius $r = \frac{\gamma_n^2}{n \log(\gamma_n)}$, $\widehat{\V}^{(0)}_0 = \emptyset$\; 
	
\textbf{Stage~1: Almost exact recovery via hypergraph spectral clustering} \\
	$\AH \leftarrow$ $\HH \HH^\top - \mathbf{D}$; \ $\AHG \leftarrow$ trim the rows and columns in $\AH$ that correspond to $v \notin \Gamma$\;
	$\AHGk \leftarrow$ rank-$k$ approximation of $\AHG$\;

	$\Psi \leftarrow$ a subset  of nodes that contains $\lceil \log n \rceil$ random samples (with replacement) from $\Gamma$\;
	$\B_v \leftarrow$ $\{u \in \Gamma: \Vert (\AHGk)_u - (\AHGk)_v \Vert_2^2 \le r \}$, for all $v \in \Psi$\;
  
	\For{$j = 1$ $\KwTo$ $k$}{
	$v_j^{\ast} \leftarrow \argmax_{v \in \Psi} |\B_v \setminus (\cup_{l=0}^{j-1} \widehat{\V}^{(0)}_l)|$\;
	$\widehat{\V}^{(0)}_j \leftarrow$ $\B_{v_j^{\ast}} \setminus (\cup_{l=0}^{j-1} \widehat{\V}^{(0)}_l)$
	}
	
	\For{$v \in \Gamma \setminus (\cup_{j=1}^k \widehat{\V}^{(0)}_j)$}{
	$j_{\ast} \leftarrow \argmin_{j} \Vert (\AHGk)_v - (\AHGk)_{v_j^{\ast}} \Vert_2^2$\;
	$\widehat{\V}^{(0)}_{j_{\ast}} \leftarrow \widehat{\V}^{(0)}_{j_{\ast}} \cup \{v\}$
	}

	Randomly assign each $v \in [n] \setminus \Gamma$ to one community in $\{\widehat{\V}^{(0)}_i\}_{i\in[k]}$
	
	Output the initial estimate $\widehat{\Z}^{(0)}$ based on $\{\widehat{\V}^{(0)}_i\}_{i\in[k]}$
	
	\textbf{Stage 2: Local refinement steps} \\	
	\For{$v \in [n]$}{
	$\hZv \triangleq \argmax_{i \in [k]} \PP\left(Z_v = i \big|G_2 = g_2, \hZZ = \hzz  \right)$\;
	}
	
\Output{Final estimate $\widehat{\Z} = (\widehat{Z}_1, \widehat{Z}_2, \ldots, \widehat{Z}_n)$}
	\caption{{\sc The two-stage algorithm}}
	\label{algorithm:1}
\end{algorithm}

\subsection{Almost exact recovery via hypergraph spectral clustering (Stage~1)}

The main focus of this subsection is the sub-HSBM $G_1 = ([n], \EE_1)$. We apply a hypergraph spectral clustering method on $G_1$ to obtain an initial estimate of the ground-truth community vector $\Z$, denoted by $\widehat{\Z}^{(0)} = (\widehat{Z}^{(0)}_1, \ldots, \widehat{Z}^{(0)}_n)$.  

Let $\HH = [\mathrm{H}_{ve}]$ be the $n \times \binom{n}{d}$ binary incidence matrix corresponding to $G_1$ such that each entry $\mathrm{H}_{ve} = 1$ if the hyperedge $e \in \EE_1$ and $e$ contains node $v$, and $\mathrm{H}_{ve} = 0$ otherwise. Note that there is an one-to-one mapping between $\HH$ and the observed adjacency tensor $\A$, thus one can obtain $\HH$ from $\A$.  For each node $v \in [n]$, its degree (in $G_1$) is denoted by 
$d_v \triangleq \sum_{e \in \mathcal{E}} \mathrm{H}_{ve}$. Let $\mathbf{D} = \mathrm{diag}(d_1, \ldots, d_n)$ be an $n \times n$ diagonal matrix that represents the degrees of the $n$ nodes. We then define the \emph{hypergraph Laplacian} as 
\begin{align}
\AH \triangleq \HH \HH^\top - \mathbf{D}, \label{eq:L}
\end{align}
where $\AH$ is an $n \times n$ matrix and the $(i,j)$-entry represents the number of hyperedges that contain both node $i$ and node $j$.  To ensure a good performance of the hypergraph spectral clustering method, one typically needs to remove a small fractions of nodes that have significantly higher degrees~\cite{chien2019minimax}  than the average. Thus, we define the set of ``good'' nodes that have degree no larger than a certain threshold $\tau$ as 
\begin{align}
\Gamma \triangleq \{v \in [n]: d_v \le \tau \},
\end{align}
where $\tau$ is set to be $CQ_{\max}\gamma_n$ for some large constant $C > 0$, such that $\tau$ is much larger than the expected degree of every node. 

We apply Stage~1 of Algorithm~\ref{algorithm:1} (lines $2-16$) to obtain an almost exact recovery of the $k$ communities. Initially, we calculate the hypergraph Laplacian $\AH$, and then ``trim'' the rows and columns in $\AH$ that correspond to nodes that do not belong to $\Gamma$. Specifically, for each of the $n$ nodes $v \in [n]$, if $v \notin \Gamma$, we replace all the entries in the $v$-th row and $v$-th column of $\AH$ by all zeros. This yields the \emph{trimmed hypergraph Laplacian} $\AHG$. In addition, we also perform an singular value decomposition (SVD) on $\AHG$ to obtain the optimal rank-$k$ approximation $\AHGk$, i.e., $\mathbf{L}_{\Gamma}^{(k)} = \sum_{i=1}^k \sigma_i \mathbf{u}_i \mathbf{v}_i^T$ where $\sigma_1 \ge \sigma_2 \ge \cdots \ge \sigma_k$ are the largest $k$ singular values, and $\mathbf{u}_i$ and $\mathbf{v}_i$ are the corresponding singular vectors of $\mathbf{L}_{\Gamma}$.

We then perform a clustering algorithm (lines $5-16$) on the columns of $\AHGk$, i.e., the set of column vectors $\{(\AHGk)_v\}_{v \in \Gamma}$. An example of our clustering algorithm is illustrated in Fig.~\ref{fig:clustering}. We first randomly select $\lceil \log n \rceil$ nodes from $\Gamma$ (with replacement) as \emph{reference nodes}, and it can be shown (in Lemma~\ref{lemma:select} below) that each community contains at least one reference node with high probability. This set of reference nodes is denoted by $\Psi$. For each node $v \in \Psi$, we construct a ball $\B_v$ with center $v$ and radius $r \triangleq \frac{\gamma_n^2}{n \log(\gamma_n)}$ which includes all the neighboring nodes (i.e., the nodes in $\B_v$).  Among $\{\B_v\}_{v \in \Psi}$, we find the one that has the largest cardinality, declare $v_1^\ast \triangleq \argmax_{v \in \Psi} |\B_v|$, and set the largest community $\widehat{\V}^{(0)}_1$ to be $\B_{v_1^\ast}$.  To find the second largest community, we remove all the nodes in $\widehat{\V}^{(0)}_1$ and then follow a similar procedure to find the ball with the largest cardinality. That is, we declare $v_2^\ast \triangleq \argmax_{v \in \Psi} |\B_v \setminus \widehat{\V}^{(0)}_1|$, and set the second largest community $\widehat{\V}^{(0)}_2$ to be $\B_{v_2^\ast} \setminus \widehat{\V}^{(0)}_1$. By repeating this procedure for $2 \le j \le k$, we obtain $k$ estimated communities $\widehat{\V}^{(0)}_1, \widehat{\V}^{(0)}_2, \ldots, \widehat{\V}^{(0)}_k$ (lines $7-10$). Furthermore, we assign the nodes belonging to $\Gamma \setminus (\cup_{j \in [k]}\widehat{\V}^{(0)}_j)$ to their nearest communities (lines $11-14$), and the nodes that do not belong to $\Gamma$ to each community randomly (line $15$). Finally, for each node $v \in \widehat{\V}^{(0)}_j$ (for all $j \in [k]$), we set $\widehat{Z}^{(0)}_v = j$.  

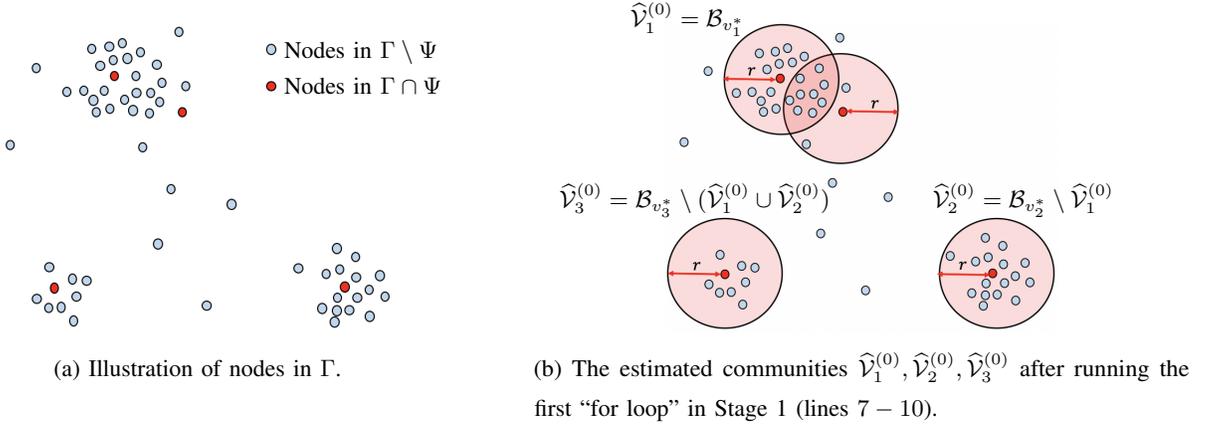
\begin{figure}[t]
\begin{center}
	\begin{subfigure}[t]{0.48\textwidth}
	\centering
		\begin{overpic}[width=.6\textwidth]{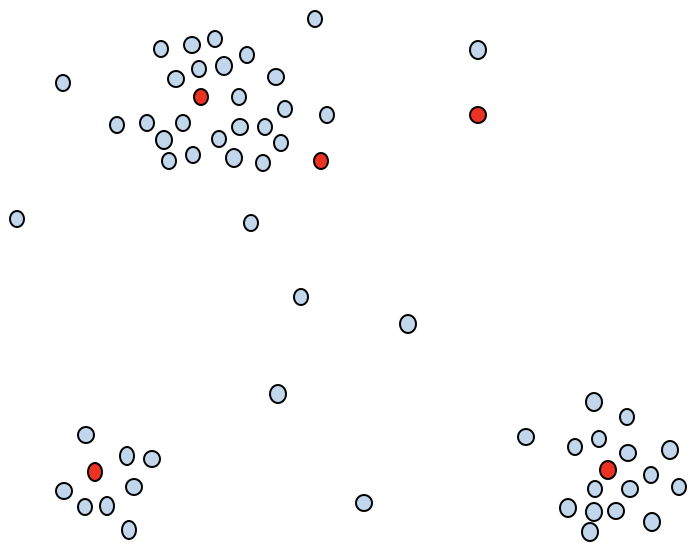}
					\put(72,70){\small Nodes in $\Gamma \setminus \Psi$}
					\put(72,61){\small Nodes in $\Gamma \cap \Psi$}
		\end{overpic}
		\caption{Illustration of nodes in $\Gamma$.}
		\label{fig:1}
	\end{subfigure}
	\begin{subfigure}[t]{0.48\textwidth}
	\centering
		\begin{overpic}[width=.6\textwidth]{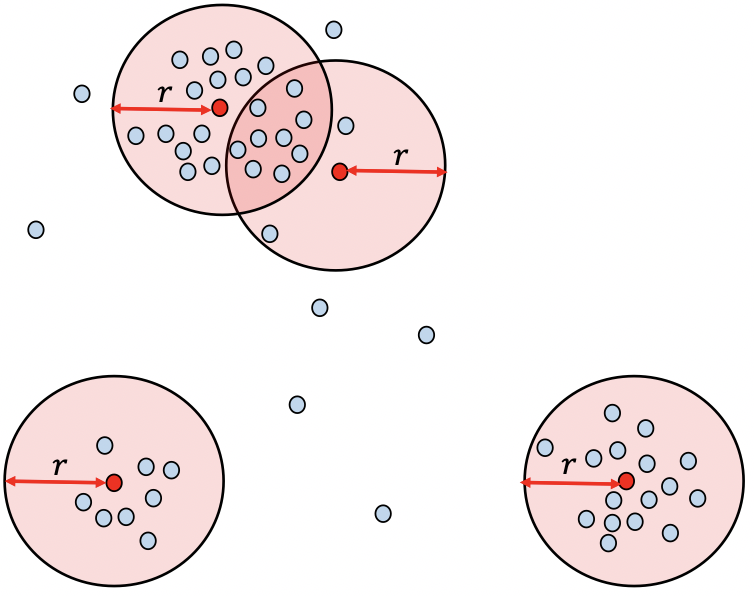}
			\put(-9,78){\small $\widehat{\V}^{(0)}_1 = \B_{v_1^{\ast}}$}
			\put(68,32){\small$\widehat{\V}^{(0)}_2 = \B_{v_2^{\ast}} \setminus \widehat{\V}^{(0)}_1$}
			\put(-27,32){\small $\widehat{\V}^{(0)}_3 = \B_{v_3^{\ast}} \setminus (\widehat{\V}^{(0)}_1\cup \widehat{\V}^{(0)}_2)$}
			\vspace{.1in}
		\end{overpic}
		
		\caption{The estimated communities $\widehat{\V}^{(0)}_1, \widehat{\V}^{(0)}_2, \widehat{\V}^{(0)}_3$ after running the first ``for loop'' in Stage~1 (lines $7-10$).}
		\label{fig:2}
	\end{subfigure}
	\caption{An example with $3$ communities. Each node corresponds to a column of $\AHGk$.}
	\label{fig:clustering}
	\end{center}
\end{figure}

The high-level intuition of the analysis of Stage~1 is as follows. Let $\MM \triangleq \E(\AHG)$ be the expected trimmed hypergraph Laplacian, where $\MM$ is identical to $\E(\mathbf{L})$ except that the rows and columns corresponding to nodes that do not belong to $\Gamma$ are set to zeros. Note that $\MM$ is an $n \times n$ matrix of rank at most $k$ when $\big(\p, \{Q_{\T}\}_{\T \in \TT}\big)\notin \Xi$. If nodes $u$ and $v$ are in the same cluster, we have $\MM_u = \MM_v$; otherwise they are far apart in the sense that $\Vert \MM_u - \MM_v \Vert_2^2 = \Omega(\gamma_n^2/n)$ (as shown in Lemma~\ref{lemma:M} below). On the other hand, the sum of the distances between each column $(\AHGk)_v$ and its expectation $\MM_v$ satisfies 
$$\sum_{v \in \Gamma} \Vert (\AHGk)_v - \MM_v \Vert_2^2 = \Vert \AHGk - \MM \Vert_{\mathrm{F}}^2 \stackrel{(a)}{\le} 8k\Vert \AHG - \MM \Vert_{\mathrm{op}}^2 \stackrel{(b)}{=} \mathcal{O}(\gamma_n)$$
with high probability, where (a) follows from the \emph{Eckart–Young–Mirsky theorem} and \emph{Weyl’s inequality}, and (b) is proved by leveraging the random matrix theory (see Lemma~\ref{lemma:op} for details). By a careful analysis of Stage~1, one can show that each node $v$ will be misclassified only if the distance $\Vert (\AHGk)_v - \MM_v \Vert_2^2 = \Omega(\gamma_n^2/n)$ (as proved in Section~\ref{sec:theory_stage1}). Thus, the output of Stage~1 results in a  misclassification of at most $\mathcal{O}(n/\gamma_n)$ nodes. Since $\gamma_n = \omega(1)$, the misclassified proportion $l(\Z,\Z^{(0)})$ tends to zero as $n$ tends to infinity, i.e., almost exact recovery is achieved.

\begin{remark}
{\em The number of  randomly selected reference nodes is set to be $\lceil \log n\rceil$ such that it is large enough to ensure that each community contains at least one reference node (with high probability). Since the radius $r = \gamma_n^2/(n\log\gamma_n)$ is also large enough, one can use the reference node in each community (denoted by $v$) to find most of its community members, via the ball $\B_v$ with center $v$ and radius $r$. 
As a result, we only need to compute $\Theta(n\log n)$ pairwise distances $\Vert (\AHGk)_u - (\AHGk)_v \Vert_2^2$; while some related works~\cite{chien2019minimax} use all the $n$ nodes as reference nodes and thus $\Theta(n^2)$ pairwise distances need to be computed.     } 
\end{remark}

\subsection{Local refinements (Stage 2)}
After obtaining the initial estimate $\widehat{\Z}^{(0)}$, we refine the label of each node $v \in [n]$ based on the observation of the hypergraph $G_2 = g_2$ as well as the estimated labels $\{\hZZ\}$ for the remaining nodes.   For each node $v \in [n]$, we perform a local \emph{maximum a posteriori (MAP)} estimation  as follows: 
\begin{align}
\hZv(g_2,\hzz) \triangleq \argmax_{i \in [k]} \PP\left(Z_v = i \big|G_2 = g_2, \hZZ = \hzz  \right). \label{eq:local}
\end{align}   
This leads to the final estimate $\widehat{\Z} = (\widehat{Z}_1, \ldots, \widehat{Z}_n)$ of the ground-truth community vector.  A detailed analysis of  Stage 2 is provided in Section~\ref{sec:exact}.

\begin{remark}
{\em Instead of computing the posterior probability $\PP(Z_v = i | G_2 = g_2, \widehat{\mathbf{Z}}_{\sim v}^{(0)} = \widehat{\mathbf{z}}_{\sim v}^{(0)} )$ directly, one can compute the probability $\PP(G_2 = g_2 | Z_v=i, \widehat{\mathbf{Z}}_{\sim v}^{(0)} = \widehat{\mathbf{z}}_{\sim v}^{(0)} ) \cdot  p_i$ which is proportional to the posterior probability. Note that the sub-HSBM $G_2$ is generated on the
hyperedge set $\mathcal{W}_2$ of the sub-hypergraph $F_2$ (as defined in Section~IV-A due to graph splitting). The presence or absence of each hyperedge in $\mathcal{W}_2$  can be modelled by a Bernoulli random variable whose success probability is governed by $Z_v$ and $\widehat{\mathbf{z}}_{\sim v}^{(0)}$. Hence, the probability $\PP(G_2 = g_2 | Z_v=i, \widehat{\mathbf{Z}}_{\sim v}^{(0)} = \widehat{\mathbf{z}}_{\sim v}^{(0)} )$ is essentially a product of $|\mathcal{W}_2|$ terms\footnote{Note that it suffices to focus on hyperedges in $\mathcal{W}_2$ that contains node $v$ only, since the presence or absence of other hyperedges does not depend on which community node $v$ belongs to, thus has no influence on the decision rule in~\eqref{eq:local}.} corresponding to the presence or not of hyperedges in $\mathcal{W}_2$, and each term equals either the ``success probability'' or ``one minus the success probability''  depending on whether the hyperedge appears in $g_2$. }
\end{remark}

\begin{remark} \label{remark:agnostic}
	{\em It is straightforward to improve Algorithm~\ref{algorithm:1} to an agnostic algorithm that does not require the knowledge of model parameters. Before performing the local refinement steps, one can estimate the distribution of communities $\widehat{\p} = (\widehat{p}_1, \ldots, \widehat{p}_k)$ based on the estimated community vector $\widehat{\Z}^{(0)}$, and the hyperedge probabilities $\{\widehat{Q}_{\T}\}_{\T \in \TT}$ based on both $\widehat{\Z}^{(0)}$ and the hyperedges in the sub-hypergraph $G_2$. Due to the \emph{law of large numbers} and the fact that $l(\Z,\widehat{\Z}^{(0)})=o(1)$,  these estimates $\widehat{\p}$ and $\{\widehat{Q}_{\T}\}_{\T \in \TT}$ are expected to be close to the true values $\p$ and $\{Q_{\T}\}_{\T \in \TT}$ respectively. As a result, it can be shown that running the local refinement steps in Stage 2 still yields exact recovery. Furthermore, even if the number of communities $k$ is not given \emph{a priori}, one can still apply a \emph{singular value thresholding method} (as employed in~\cite{yun2016optimal}) to the hypergraph Laplacian $\AHG$ to estimate the value of $k$ in Stage~1. }	
\end{remark}

\begin{remark}
{\em We note that the algorithm in~\cite{chien2019minimax} also relies on a hypergraph spectral clustering step plus a local refinement step. However, their hypergraph spectral clustering method is different from ours, with main distinctions described as follows. 
\begin{itemize}
	\item[--] In our work, we use the entire rank-$k$ approximation of the trimmed hypergraph Laplacian $\mathbf{L}_{\Gamma}^{(k)} = \sum_{i=1}^k \sigma_i \mathbf{u}_i \mathbf{v}_i^T$ as the input of our subsequent clustering step (lines 5-15 in Algorithm~1).   
	\item[--] The algorithm in~\cite{chien2019minimax} applies a singular value decomposition to $\mathbf{L}_{\Gamma}$ to obtain the $k$ leading singular vectors $[\mathbf{u}_1,\ldots, \mathbf{u}_k]$, and then apply their subsequent clustering step by representing each of the $n$ node by a reduced $k$-dimensional vector.
\end{itemize}
The advantage of our algorithm is that its accompanying analysis does not involve the $k$-th largest singular value of $\E(\mathbf{L})$, while the theoretical guarantee of the spectral clustering step in~\cite{chien2019minimax} depends on the $k$-th largest singular value. In the setting of~\cite{chien2019minimax} (described in Section~\ref{sec:chien}), their algorithm works well because the $k$-th largest singular value is reasonably large (as shown in~\cite[Lemma~5.1]{chien2019minimax}). However, in our more general setting, the $k$-th largest singular value is not always large enough, which prohibits the applicability of the algorithm in~\cite{chien2019minimax}. In contrast, our algorithm is applicable to a larger set of parameters, i.e., as long as the second-order degree profile condition is satisfied.
 } 
\end{remark}

\section{Theoretical Guarantees of Algorithm~\ref{algorithm:1} (Theorem~\ref{thm:achievability})} \label{sec:proof}

In this section, we prove that as long as $\min_{i,j\in [k]: i \ne j} D_+(i,j) > 1$ and $\big(\p, \{Q_{\T}\}_{\T \in \TT}\big)\notin \Xi$, applying Algorithm~\ref{algorithm:1} on the observed hypergraph $G$ ensures $l(\Z, \widehat{\Z}) = 0$ with high probability for sufficiently large $n$ (i.e., exact recovery).

First, we note that  with high probability, the size of each community $\V_j = \{v \in [n]: Z_v = j\}$ is close to $np_j$ for all $j\in[k]$.  This is stated in Lemma~\ref{lemma:size} below and can be proved by applying the Chernoff bound~

\begin{lemma} \label{lemma:size}
Fix a  constant $\delta \in (0,1)$ which can be chosen to be arbitrarily small. We say that the length-$n$ vector $\Z\in\mathcal{A}_{Z}$ (where $\mathcal{A}_{Z}$ is referred to as the typical set of $\Z$) if the communities $\{\V_j\}_{j \in [k]}$ associated with $\Z$ satisfy 
\begin{align}
	 \left(1 - n^{-\frac{1}{2}+\frac{\delta}{2}} \right) p_j n \le |\V_j| \le \left(1 + n^{-\frac{1}{2}+\frac{\delta}{2}} \right) p_j n, \quad \text{for all} \ j \in [k]. \label{eq:cluster_size}
	 \end{align}
Then, we have $\PP(\Z \in \mathcal{A}_{Z}) \ge 1 - \exp(-\Theta(n^{\delta}))$.
\end{lemma}
Therefore, one can focus on typical ground-truth community vectors $\z \in \mathcal{A}_{Z}$ in  the following analysis. 

\subsection{Theoretical guarantees of Stage~1} \label{sec:theory_stage1}

For a fixed ground-truth community vector $\z \in \mathcal{A}_{Z}$, we first introduce an artificial $d$-HSBM $\widetilde{G}_1$ which is generated with respect to the ground-truth community vector $\z$ and hyperedge probabilities $\{\frac{Q_{\T}\gamma_n}{n^{d-1}}\}_{\T \in \TT}$. Note that the generation process of $\widetilde{G}_1$ is equivalent to first generating a sub-hypergraph $F_1$ (with splitting parameter $\gamma_n/\log n$) and then generating a sub-HSBM $G_1$ on the hyperedge set of $F_1$ (with hyperedge probabilities $\{\frac{Q_{\T}\log n}{n^{d-1}}\}_{\T \in \TT}$). Thus, we investigate the misclassification proportion $l(\z, \widehat{\Z}^{(0)})$ based on the random hypergraph $\widetilde{G}_1$.

\begin{theorem}[Theoretical guarantee of Stage~1]\label{thm:weak}
Suppose the model parameters $\big(\p, \{Q_{\T}\}_{\T \in \TT}\big)\notin \Xi$. For any fixed  $\z \in \mathcal{A}_{Z}$, there exist vanishing sequences $\{\epsilon_n\}$ and $\{\eta_n\}$ (which depend on $\{\gamma_n\}$) such that with probability at least $1- \epsilon_n$ over the generation of $\widetilde{G}_1$, running Stage~1 of Algorithm~\ref{algorithm:1} ensures that $l(\z, \widehat{\Z}^{(0)}) \le \eta_n$, i.e., almost exact recovery is achieved.
\end{theorem}

Theorem~\ref{thm:weak} is proved  in the rest of this section, and in the following we  assume that $\big(\p, \{Q_{\T}\}_{\T \in \TT}\big)\notin \Xi$.

Note that with high probability over the generation of $\widetilde{G}_1$, the degrees of most nodes in $[n]$ are smaller than the threshold $\tau$, thus only a vanishing fraction of nodes (at most $\mathcal{O}(n/\gamma_n)$ nodes) is trimmed. This result is adapted from~\cite[Lemma D.3]{chien2019minimax} and stated below.

\begin{lemma}[Adapted from Lemma D.3 in~\cite{chien2019minimax}] \label{lemma:gamma}
There exists a large constant $C > 0$ such that if we set $\tau = C \Qm\gamma_n$, then with probability at least $1 - \exp(-C'n)$ (for some constant $C' > 0$), the cardinality of the set $\Gamma$ satisfies $|\Gamma| \ge n(1 - (1/\tau))$.
\end{lemma}

Furthermore, let $\VG_j \triangleq \V_j \cap \Gamma$ for all $j \in [k]$. Lemma~\ref{lemma:gamma} also implies that with high probability,  
\begin{align}
    (1-o(1)) p_j n \le |\VG_j| \le (1+o(1)) p_j n, \quad \text{for all} \ j \in [k].
\end{align}
We then focus on the remaining nodes in the set $\Gamma$. First recall that $\MM = \E(\AHG)$ is a matrix of rank at most $k$, and satisfies  $\MM_u = \MM_v$ if $u$ and $v$ belong to the same community (where $u,v \in \Gamma$). Lemma~\ref{lemma:M} below shows that the distance between the two columns $\MM_u$ and $\MM_v$ scales as $\Omega(\gamma_n^2/n)$ if $u$ and $v$ belong to different communities. 

\begin{lemma} \label{lemma:M}
Suppose $u, v \in \Gamma$. When the second-order degree profile condition is satisfied (i.e., $\big(\p, \{Q_{\T}\}_{\T \in \TT}\big)\notin \Xi$), we have
\begin{align}
    \Vert \MM_u - \MM_v \Vert_2^2 = \begin{cases} 0, &\text{if } u \text{ and }  v \  \text{belong to the same community}; \\
    \Omega(\gamma_n^2/n), &\text{otherwise}. \label{eq:M}
    \end{cases}
\end{align}
\end{lemma}
The proof of Lemma~\ref{lemma:M} can be found in Appendix~\ref{appendix:M}. Lemma~\ref{lemma:op} below shows that with high probability, both the operator norm and the Frobenius norm of the difference between $\AHG$ and $\MM$ can be appropriately upper bounded. 

\begin{lemma} \label{lemma:op}
For any constant $C'_1 > 0$, there exists some constant $C_1 > 0$ such that with probability at least $1 - n^{-C_1'}$, 
\begin{align}
\Vert \AHG - \MM \Vert_{\mathrm{op}} \le C_1\cdot \left[\sqrt{\gamma_n\Qm} + \sqrt{\tau} + \frac{\gamma_n\Qm}{\sqrt{\gamma_n\Qm}+ \sqrt{\tau}}  \right]. \label{eq:op}
\end{align}
As a result, the Frobenius norm between $\MM$ and $\AHGk$ (the rank-$k$ approximation of $\AHG$) can be upper-bounded as 
\begin{align}
    \Vert \AHGk - \MM \Vert_{\mathrm{F}}^2  \le 8kC_1^2\cdot \left[\sqrt{\gamma_n\Qm} + \sqrt{\tau} + \frac{\gamma_n\Qm}{\sqrt{\gamma_n\Qm}+ \sqrt{\tau}}  \right]^2 = \mathcal{O}(\gamma_n). \label{eq:f}
\end{align}
\end{lemma}

\begin{proof}[Proof of Lemma~\ref{lemma:op}]
The proof of~\eqref{eq:op} can be adapted from~\cite[Lemma D.4]{chien2019minimax}, so our main focus is on the proof of~\eqref{eq:f}.  Note that 
\begin{align}
    \Vert \AHGk - \MM \Vert_{\mathrm{F}}^2 &\le 2k \cdot \Vert \AHGk - \MM \Vert_{\mathrm{op}}^2 \label{eq:1} \\
    & \le 4k \cdot \Vert \AHGk - \AHG \Vert_{\mathrm{op}}^2 + 4k \cdot \Vert \AHG - \MM \Vert_{\mathrm{op}}^2 \label{eq:2} \\
    & \le 4k \cdot \sigma_{k+1}^2 + 4k \cdot \Vert \AHG - \MM \Vert_{\mathrm{op}}^2 \label{eq:3} \\
    &\le 8k \cdot \Vert \AHG - \MM \Vert_{\mathrm{op}}^2, \label{eq:4}
\end{align}
where $\sigma_{k+1}$ is the $(k+1)$-th largest singular value of $\AHG$. Eqn.~\eqref{eq:1} holds since the rank of $\AHGk - \MM$ is at most $2k$, and $\Vert \mathbf{X} \Vert_{\mathrm{F}}^2 \le r \Vert \mathbf{X} \Vert_{\mathrm{op}}^2$ for any matrix $\mathbf{X}$ of rank $r$. Eqn.~\eqref{eq:2} holds since $\Vert \AHGk - \MM \Vert_{\mathrm{op}}  \le \Vert \AHGk - \AHG \Vert_{\mathrm{op}} + \Vert \AHG - \MM \Vert_{\mathrm{op}}$ and $2\Vert \AHGk - \AHG \Vert_{\mathrm{op}}\cdot  \Vert \AHG - \MM \Vert_{\mathrm{op}} \le \Vert \AHGk - \AHG \Vert_{\mathrm{op}}^2 + \Vert \AHG - \MM \Vert_{\mathrm{op}}^2$. Eqn.~\eqref{eq:3} follows from the Eckart–Young–Mirsky theorem~\cite{eckart1936approximation}, while Eqn.~\eqref{eq:4} is due to Weyl's inequality\footnote{Let $\mathbf{A}$ and $\mathbf{B}$ be $n \times n$ matrices, and their singular values are respectively denoted by $\{\sigma_{i}(\mathbf{A})\}_{i\in[n]}$ and $\{\sigma_{i}(\mathbf{B})\}_{i\in[n]}$, both in decreasing orders. Weyl's inequality states that for every $i \in [n]$, $|\sigma_{i}(\mathbf{A}) - \sigma_{i}(\mathbf{B})| \le \Vert \mathbf{A} - \mathbf{B} \Vert_{\mathrm{op}}$. }~\cite{weyl1912asymptotische}. Combining Eqns~\eqref{eq:op} and~\eqref{eq:4}, it is then clear that 
\begin{align}
    \Vert \AHGk - \MM \Vert_{\mathrm{F}}^2 &\le 2k \cdot \Vert \AHGk - \MM \Vert_{\mathrm{op}}^2 \le 8kC_1^2\cdot \left[\sqrt{\gamma_n\Qm} + \sqrt{\tau} + \frac{\gamma_n\Qm}{\sqrt{\gamma_n\Qm}+ \sqrt{\tau}}  \right]^2 = \mathcal{O}(\gamma_n). 
\end{align}
This completes the proof of Lemma~\ref{lemma:op}.
\end{proof}

\noindent{\bf Analysis of Stage~1:}
In the following, we show that when Lemmas~\ref{lemma:gamma}--\ref{lemma:op} hold, running Stage~1 ensures that $l(\z, \widehat{\Z}^{(0)})\to 0$ with high probability. Our analysis is inspired by~\cite{yun2016optimal} for SBMs, but is adapted to our general $d$-HSBM problem.
We first partition the nodes in $\Gamma$ as follows. Recall that the radius $r$ is set to be $\gamma_n^2/(n \log(\gamma_n))$, and let  
\begin{align}
    \Ii_j &\triangleq \{v \in \VG_j: \Vert (\AHGk)_v - \MM_v \Vert_2^2 \le r/4 \}, \quad \text{for all} \ j \in [k], \\
    \Io_j &\triangleq \{v \in \VG_j: \Vert (\AHGk)_v - \MM_v \Vert_2^2 \le 4r \}, \ \ \quad \text{for all} \ j \in [k],
\end{align}
and the remaining nodes in $\Gamma$ belong to $\U \triangleq \Gamma \setminus \{\cup_{j \in [k]}\Io_j \}$. Note that these sets have the following properties: 
\begin{enumerate}[label=(\roman*)]
    \item For all $j \in [k]$, most of the nodes $v\in\VG_j$ are such that $(\AHGk)_v$ is $\frac{r}{4}$-close to its expectation $\MM_{v}$, since
    \begin{align}
    |\VG_j \setminus \Ii_j| \le \frac{\Vert \AHGk - \MM \Vert_{\mathrm{F}}^2}{r/4} = \mathcal{O}\left(\frac{\log(\gamma_n)}{\gamma_n}n\right) = o(n).
    \end{align}
    Thus, $|\Ii_j| \ge |\VG_j \cap \Ii_j| = |\VG_j| - |\VG_j \setminus \Ii_j| \ge p_j n (1-o(1))$.
    
    \item Similar to (i), we have $|\U| \le  |\Gamma \setminus (\cup_{j\in[k]}\Ii_j)| \le \Vert \AHGk - \MM \Vert_{\mathrm{F}}^2/(r/4) = \mathcal{O}\left(\frac{\log(\gamma_n)}{\gamma_n}n\right) = o(n)$.
    
    \item If node $v \in \U \cap \Psi$, then $\B_v \cap (\cup_{j\in[k]}\Ii_j) = \emptyset$, and as a result of (i), we have $|\B_v| = \mathcal{O}\left(\frac{\log(\gamma_n)}{\gamma_n}n\right) = o(n)$.   
    
    \item If node $v \in \Ii_j \cap \Psi$, then $\Ii_j \subseteq \B_v$. 
    
    \item If nodes $u \in \Io_i$ and $v \in \Io_j$ such that  $i \ne j$, then $T_u \cap \B_v = \emptyset$.
    
    \item For all $j \in [k]$, $|\Io_j| \le p_j n (1+o(1))$. This is because $|\Ii_j| \ge p_j n (1-o(1))$ from (i), and thus  
    $$|\Io_j| \le n - |\Gamma| - \sum_{j' \ne j} |\Io_{j'}| \le n - \sum_{j' \ne j} |\Ii_{j'}| \le n - (p_{j'}n(1-o(1))) \le p_j n (1+o(1)).$$ 
\end{enumerate}

Lemma~\ref{lemma:select} below states that with high probability, each community contains at least one reference node. 

\begin{lemma} \label{lemma:select}
With probability $1 - o(1)$ over the selection of $\Psi$, we have $|\Ii_j \cap \Psi| \ge 1$ for all $j \in [k]$.
\end{lemma}
\begin{proof}
For any $j \in [k]$, the probability that a randomly selected reference node does not belong to $\Ii_j$ is $1 - (|\Ii_j|/n)$, thus the probability that there exists at least one reference node belongs to $\Ii_j$ is 
\begin{align}
	1 - \left(1 - (|\Ii_j|/n)\right)^{\log n} = 1 - \exp(-\Theta(n/\log n)).
\end{align}  
Taking a union bound over all $j \in [k]$, we complete the proof.
\end{proof}

With these properties, we are able to show that running Stage~1  yields almost exact recovery. For ease of presentation, in the following we focus mainly  on the case when $p_1 > p_2 > \ldots > p_k$ (such that the community sizes satisfy $|\VG_1| > |\VG_2| > \ldots > |\VG_k|$ with high probability). The analysis can be easily generalized to the case when $p_i = p_j$ for some $i \ne j$, but it requires more cumbersome notations (such as permutations) which makes the subsequent analysis more difficult to understand. 

\noindent{\underline{The ``for loop'' in lines $7-10$:}} When $j = 1$, we will show that $v_1^{\ast} \in \Io_1$ and $|\B_{v_1^{\ast}}| \ge np_1(1-o(1))$. From Lemma~\ref{lemma:select}, there exists a node $v_1 \in \Ii_1 \cap \Psi$, and its corresponding set $\B_{v_1}$ is a superset of $\Ii_1$ (due to (iv)). Thus, we have $|\B_{v_1}| \ge |\Ii_1| \ge p_1n(1-o(1))$, where the last inequality is due to (i). As $|\B_{v_1^{\ast}}|$ is at least $|\B_{v_1}|$, we obtain that $|\B_{v_1^{\ast}}| \ge p_1n(1-o(1))$. To prove that $v_1^{\ast} \in \Io_1$, one can verify that 
\begin{enumerate}
    \item for any $v \in \U \cap \Psi$,  $|\B_v| = o(n) < |\B_{v_1^{\ast}}|$ by (iii); 
    \item for any $v \in \Io_j \cap \Psi$ where $j \ne 1$, $|\B_v| \le |\Io_j| + |\U| \le np_j (1+o(1)) + o(n) < |\B_{v_1^{\ast}}|$ (due to (vi) and the fact that $p_1 > p_j$).
\end{enumerate}

For $2 \le j \le k$, we will show that $v_j^{\ast} \in \Io_j$ and $|\B_{v_j^{\ast}}| \ge np_j(1-o(1))$. Similar to the analysis above, there exists a node $v_j \in \Ii_j \cap \Psi$ and $\B_{v_j} \supseteq \Ii_j$. From (v) we know that $|\B_{v_j} \setminus (\cup_{l=0}^{j-1} \widehat{\V}^{(0)}_l)| = |\B_{v_j}|$, thus $|\B_{v_j} \setminus (\cup_{l=0}^{j-1} \widehat{\V}^{(0)}_l)| \ge |\Ii_j| \ge np_j(1-o(1))$. As $|\B_{v_j^{\ast}} \setminus (\cup_{l=0}^{j-1} \widehat{\V}^{(0)}_l)|$ is at least $|\B_{v_j} \setminus (\cup_{l=0}^{j-1} \widehat{\V}^{(0)}_l)|$, we obtain that $|\B_{v_j^{\ast}}| \ge p_jn(1-o(1))$. To prove that $v_j^{\ast} \in \Io_j$, one can verify that 
\begin{enumerate}
    \item for any $v \in \U \cap \Psi$,  $|\B_v \setminus (\cup_{l=0}^{j-1} \widehat{\V}^{(0)}_l)| = o(n)$;
    \item for any $v \in \Io_{j'} \cap \Psi$ where $j' > j$, $|\B_v \setminus (\cup_{l=0}^{j-1} \widehat{\V}^{(0)}_l)| \le |\B_v| \le |\Io_{j'}| + |\U| \le np_{j'} (1+o(1)) + o(n) < |\B_{v_{j}^{\ast}}|$ (due to (vi) and the fact that $p_j > p_{j'}$);
    \item for any $v \in \Io_{j'} \cap \Psi$ where $j' < j$, we have $|\B_v \setminus (\cup_{l=0}^{j-1} \widehat{\V}^{(0)}_l)| \le |\B_v \setminus \widehat{\V}^{(0)}_{j'}| \le |\Io_j| + |\U| - |\B_{v_{j'}^{\ast}}| = o(n)$, since $|\B_{v_{j'}^{\ast}}| \ge np_{j'}(1-o(1))$.
\end{enumerate}
Therefore, under the assumption $p_1 > p_2 > \cdots > p_k$, we have $v_j^{\ast} \in \Io_j$ for all $j \in [k]$.\footnote{Without this assumption, we have that there exists a permutation $\pi$ on $[k]$ such that $v_j^{\ast} \in \Io_{\pi(j)}$ for all $j \in [k]$.} This implies that all the elements in $\{v_{j}^{\ast}\}_{j\in[k]}$ are far from each other---specifically, every pair of nodes $(v_i^{\ast}, v_j^{\ast})$ satisfies 
\begin{align}
\Vert (\AHGk)_{v_i^{\ast}} - (\AHGk)_{v_j^{\ast}} \Vert_2^2 &\ge \frac{1}{2}\Vert \MM_{v_i^{\ast}} - (\AHGk)_{v_j^{\ast}} \Vert_2^2 - \Vert \MM_{v_i^{\ast}} - (\AHGk)_{v_i^{\ast}} \Vert_2^2  \notag \\
&\ge \frac{1}{2}\left(\frac{1}{2} \Vert \MM_{v_j^{\ast}} - \MM_{v_i^{\ast}} \Vert_2^2  - \Vert \MM_{v_j^{\ast}} - (\AHGk)_{v_j^{\ast}} \Vert_2^2 \right) - \Vert \MM_{v_i^{\ast}} - (\AHGk)_{v_i^{\ast}} \Vert_2^2 = \Omega(\gamma_n^2/n),
\end{align}
since $\Vert (\AHGk)_{v_i^{\ast}} - \MM_{v_i^{\ast}} \Vert_2^2 \le 4r$, $\Vert (\AHGk)_{v_j^{\ast}} - \MM_{v_j^{\ast}} \Vert_2^2 \le 4r$, and $\Vert \MM_{v_j^{\ast}} - \MM_{v_i^{\ast}} \Vert_2^2 = \Omega(\gamma_n^2/n)$ by Lemma~\ref{lemma:M}. If a node $v \in \VG_i$ is misclassified to $\widehat{\V}^{(0)}_j$ (where $i\ne j$) in the first ``for loop'', then it must be close to the center of $\widehat{\V}^{(0)}_j$, i.e., $\Vert (\AHGk)_v - (\AHGk)_{v_j^{\ast}} \Vert_2^2 \le r$. Thus, we have
\begin{align}
\Vert (\AHGk)_v - \MM_v \Vert_2^2 &\ge \frac{1}{2} \Vert \MM_v - \MM_{v_j^{\ast}} \Vert_2^2 - \Vert (\AHGk)_v - \MM_{v_j^{\ast}} \Vert_2^2 \notag \\
&\ge \frac{1}{2} \Vert \MM_v - \MM_{v_j^{\ast}} \Vert_2^2 - 2\Vert (\AHGk)_v - (\AHGk)_{v_j^{\ast}} \Vert_2^2 - 2\Vert  (\AHGk)_{v_j^{\ast}} - \MM_{v_j^{\ast}}  \Vert_2^2 \notag \\
&= \Omega(\gamma_n^2/n). \label{eq:must1}
\end{align} 

\noindent{\underline{The ``for loop'' in lines $11-14$:}} Consider a specific node $v \in \Gamma \setminus (\cup_{j=1}^k \widehat{\V}^{(0)}_j)$. If $v \in \VG_i$ is misclassified to $\widehat{\V}^{(0)}_j$ (where $i \ne j$) in the second ``for loop'', then it must be closer to the center of $\widehat{\V}^{(0)}_j$ than the center of $\widehat{\V}^{(0)}_i$, i.e., $\Vert (\AHGk)_v - (\AHGk)_{v_j^{\ast}} \Vert_2^2 \le \Vert (\AHGk)_v - (\AHGk)_{v_i^{\ast}} \Vert_2^2$. Since the two centers are far from each other, one can show that $v$ is far from the center of $\widehat{\V}^{(0)}_i$, i.e., 
\begin{align}
\Omega(\gamma_n^2/n) \le \Vert (\AHGk)_{v_i^{\ast}} - (\AHGk)_{v_j^{\ast}} \Vert_2^2 &\le 2\Vert (\AHGk)_{v_i^{\ast}} - (\AHGk)_{v} \Vert_2^2 + 2\Vert (\AHGk)_{v_j^{\ast}} - (\AHGk)_{v} \Vert_2^2 \le 4\Vert (\AHGk)_{v_i^{\ast}} - (\AHGk)_{v} \Vert_2^2.
\end{align}
Thus, node $v$ must satisfy 
\begin{align}
\Vert (\AHGk)_v - \MM_v \Vert_2^2 \ge \frac{1}{2} \Vert (\AHGk)_v - (\AHGk)_{v_i^{\ast}} \Vert_2^2 - \Vert \MM_v - (\AHGk)_{v_i^{\ast}} \Vert_2^2 = \Omega(\gamma_n^2/n). \label{eq:must2}
\end{align}

Combining Eqns.~\eqref{eq:must1} and~\eqref{eq:must2}, we conclude that for any node $v \in \Gamma$, if it is misclassified to another cluster, it must satisfy $\Vert (\AHGk)_v - \MM_v \Vert_2^2 = \Omega(\gamma_n^2/n)$.  Since $\sum_{v \in \Gamma} \Vert (\AHGk)_v - \MM_v  \Vert_2^2 = \Vert \AHGk - \MM \Vert_{\mathrm{F}}^2 = \mathcal{O}(\gamma_n)$ (by Lemma~\ref{lemma:op}), we know that the number of misclassified nodes in $\Gamma$ is at most $\mathcal{O}(n/\gamma_n)$.  Taking into account the number of nodes that do not belong to $\Gamma$ (which also scales as $\mathcal{O}(n/\gamma_n)$ by Lemma~\ref{lemma:gamma}), we complete the proof of Theorem~\ref{thm:weak}.

\subsection{Theoretical guarantees of Stage 2} \label{sec:exact}
From Theorem~\ref{thm:weak} we know that for a fixed ground-truth community vector $\z \in \mathcal{A}_{Z}$, running Stage~1 on $\widetilde{G}_1$ ensures that $l(\z, \widehat{\Z}^{(0)}) \le \eta_n$ with probability at least $1 - \epsilon_n$. In the following, we show that the hypergraph spectral clustering method  does not only work well on $\widetilde{G}_1$, but also works well on the sub-HSBM $G_1$ (which is generated on the fixed sub-hypergraph $F_1$) with high probability over the graph splitting process.

\begin{definition} \label{def:subgraph}
Let $(f_1, f_2)$ be the realizations of the sub-HSBMs $(F_1, F_2)$, and consider a fixed ground-truth community vector $\z$. 
\begin{itemize}
\item We say $f_1$ is a \emph{good realization of the first sub-hypergraph} with respect to $\z$ (denoted by $f_1 \in \G^{\z}_{1}$) if the probability that ``running Stage~1 on $G_1$ (which depends on $f_1$) ensures $l(\z, \widehat{\Z}^{(0)}) \le \eta_n$'' is at least $1- \sqrt{\epsilon_n}$, i.e., $\PP_{G_1}(l(\z, \widehat{\Z}^{(0)}) \le \eta_n) \ge 1- \sqrt{\epsilon_n}$.

\item We say $f_2$ is a \emph{good realization of the second sub-hypergraph} with respect to $\z$ (denoted by $f_2 \in \G^{\z}_{2}$) if for every node $v \in [n]$ and every community assignment $\m \in \M$, the number of hyperedges in $f_2$ that includes node $v$ and other $d-1$ nodes with community assignment $\m$ (denoted by $D_{v,\m}$) satisfies  $\left(1 - (3\gamma_n/\log n) \right) R_{\m} \le D_{v,\m} \le (1 + n^{-\frac{1}{2}+\frac{\delta}{2}}) R_{\m}$.
\end{itemize}
\end{definition} 

\begin{lemma}  \label{lemma:realization}
Suppose $\z \in \mathcal{A}_{Z}$.
With probability at least $1 - 2\sqrt{\epsilon_n}$, the randomly generated sub-hypergraphs $(F_1, F_2)$ satisfy $F_1 \in \G^{\z}_{1}$ and $F_2 \in \G^{\z}_{2}$ simultaneously.
\end{lemma}
\begin{proof}
	See Appendix~\ref{appendix:realization}.
\end{proof}
Due to Lemma~\ref{lemma:realization}, one can then focus on a specific typical ground-truth community vector $\z \in \mathcal{A}_{Z}$ and good realizations of the sub-hypergraphs $f_1 \in \G^{\z}_{1}$ and $f_2 \in \G^{\z}_{2}$ in the following. Also, one can suppose that the initial estimate $\widehat{\z}^{(0)}$ (after Stage~1) satisfies $l(\z, \widehat{\z}^{(0)}) \le  \eta_n$, since running Stage~1 on $G_1$ (which depends on the good realization $f_1$) ensures $l(\z, \widehat{\Z}^{(0)}) \le \eta_n$ with high probability. Without loss of generality, we further assume that the minimum value of $l(\z, \widehat{\z}^{(0)}) = \min_{\pi \in \mathcal{S}_k}\frac{1}{n} \sum_{v \in[n]} \mathbbm{1}\{\widehat{z}^{(0)}_v \ne \pi(z_v) \}$ is achieved by the identity permutation in $\mathcal{S}_k$, and this assumption helps us to remove the presence of the  permutation in the following analysis. 

Note that the local MAP estimation in~\eqref{eq:local} for each node $v \in [n]$ can be alternatively represented as 
\begin{align}
\hZv(g_2,\hzz) &= \argmax_{i \in [k]} \PP\left(Z_v = i \big|G_2 = g_2, \hZZ = \hzz  \right) \\
&= \argmax_{i \in [k]} \PP\left(G_2 = g_2 \big|Z_v = i, \hZZ = \hzz  \right) \cdot p_i.
\end{align}
Thus, for a fixed node $v \in \V_i$ (i.e., $z_v = i$) for some $i \in [k]$, the probability that it is misclassified to a different community can be bounded as follows:
\begin{align}
\PP\left(\hZv(G_2,\hzz) \ne z_v  \right) &= \PP_{G_2}\left( \argmax_{i' \in [k]} \PP\left(G_2 \big|Z_v = i', \hzz \right) \cdot p_{i'} \ne i \right) \\
&= \PP_{G_2}\left( \exists j \ne i : \PP\left(G_2 \big|Z_v = j, \hzz \right)\cdot p_j \ge \PP\left(G_2 \big|Z_v = i, \hzz \right)\cdot p_i \right) \\
&\le \sum_{j \ne i} \PP_{G_2}\left( \PP\left(G_2 \big|Z_v = j, \hzz \right)\cdot p_j \ge \PP\left(G_2 \big|Z_v = i, \hzz \right) \cdot p_i \right) \\
&=\sum_{j \ne i} \sum_{g_2} \PP\left(g_2 \big|Z_v = i,  \zsv \right) \times \mathbbm{1}\left\{\PP\left(g_2 \big|Z_v = j, \hzz \right)\cdot p_j \ge \PP\left(g_2 \big|Z_v = i, \hzz \right)\cdot p_i \right\}. \label{eq:measure}
\end{align}
Note that the above error probability depends only on the hyperedges in $G_2$ (which is generated on the sub-hypergraph $f_2$) that contains node $v$. Recall that the number of hyperedges in $f_2$ that comprise node $v$ and other $d-1$ nodes with community assignment $\m$ is $D_{v,\m}$ (as defined in Definition~\ref{def:subgraph}), thus these hyperedges can be represented by Bernoulli random variables $\{X^{\m}_{a} \}_{a=1}^{D_{v,\m}}$, where $X_a^{\m} = 1$ means the presence of this hyperedge in $G_2$. Taking all possible $\m \in \M$ into account, we know that the hyperedges in $f_2$ that contains node $v$ can be represented by the collection of Bernoulli random variables $\big\{\{X^{\m}_{a} \}_{a=1}^{D_{v,\m}} \big\}_{\m \in \M}$. We also denote $D_v \triangleq \sum_{\m \in \M} D_{v,\m}$ as the total number of hyperedges in $f_2$ that contains $v$, where $D_{v}$ scales as $\Theta(n^{d-1})$ by Lemma~\ref{lemma:realization}.  

To further upper bound~\eqref{eq:measure}, we would like to substitute the terms $\PP(g_2 \big|Z_v = i,  \hzz)$ and $\PP(g_2 \big|Z_v = j,  \hzz)$ in the indicator function by $\PP\left(g_2 \big|Z_v = i,  \zsv \right)$ and $\PP\left(g_2 \big|Z_v = j,  \zsv \right)$ respectively; thus we can get rid of the fact that $\hzz$ is not exactly the same as $\zsv$. We now consider the ratio between $\PP(G_2 \big|Z_v = j,  \hzz)$ and $\PP\left(G_2 \big|Z_v = j,  \zsv \right)$ for an arbitrary $j \in [k]$. Due to the independence of hyperedges $\big\{\{X^{\m}_{a} \}_{a=1}^{D_{v,\m}} \big\}_{\m \in \M}$ in $G_2$, we have 
\begin{align}
\frac{\PP\left(G_2 \big|Z_v = j,  \hzz \right)}{\PP\left(G_2 \big|Z_v = j,  \zsv \right)} = \prod_{\m \in \M}\prod_{a=1}^{D_{v,\m}} \frac{\PP\left(X^{\m}_a \big|Z_v = j,  \hzz \right)}{\PP\left(X^{\m}_{a} \big|Z_v = j,  \zsv \right)} \label{eq:ratio}
\end{align} 
Conditioned on $Z_v = j$ and the ground truth $\zsv$, the hyperedge $X^{\m}_a$ includes node $v \in \V_j$ and  $d-1$ other nodes with community assignment $\m$, thus $\PP(X^{\m}_a = 1\big|Z_v = j,  \zsv) = \frac{Q_{\m \oplus j} \log n}{n^{d-1}}$. On the other hand, conditioned on $Z_v = j$ and the estimated labels $\hzz$, the hyperedge $X^{\m}_a$ is considered to include node $v \in \V_j$, but the community assignment of the other $d-1$ nodes, denoted by $\m' \in \M$, may not be equal to  $\m$, and we have $\PP(X^{\m}_a = 1\big|Z_v = j,  \hzz) = \frac{Q_{\m' \oplus j} \log n}{n^{d-1}}$. Here, $\m' = \m$ if all the $d-1$ nodes (except for $v$) in this hyperedge are not misclassified in $\hzz$, while $\m' \ne \m$ otherwise. For hyperedges such that their community assignments $\m'$ (under $\hzz$) equal $\m$, they are cancelled out in~\eqref{eq:ratio} since the distributions are invariant conditioned on either $\hzz$ or $\zsv$.  It then remains to focus on the hyperedges that contain at least one misclassified node in $\hzz$. Since $l(\z,\widehat{\z}^{(0)}) \le \eta_n$, the number of such hyperedges is at most $\eta_n n\cdot \binom{n-2}{d-2} \triangleq D'_v$, where $D'_v$ scales as $\mathcal{O}(\eta_n n^{d-1})$. For ease of presentation, these $D'_v$ Bernoulli random variables in $\big\{\{X^{\m}_{a} \}_{a=1}^{D_{v,\m}} \big\}_{\m \in \M}$ are alternatively relabelled as $\{Y_{a}\}_{a=1}^{D'_v}$.


\begin{definition}
Let  $\{c_n\}$ be a vanishing sequence that satisfy $c_n \log(c_n /\eta_n) = \omega(1)$. We say $G_2 \in \G_v$ if the random variables $\{Y_a\}_{a=1}^{D'_v}$ satisfies $\sum_{a=1}^{D'_v} Y_a \le c_n \log n$, and $G_2 \notin \G_v$ otherwise.
\end{definition} Lemma~\ref{lemma:approximation} below states that $G_2 \in \G_v$ with high probability (conditioned on the ground truth $Z_v = i$ and $\zsv$), and for every realization $g_2 \in \G_v$, the ratio between $\PP(g_2 |Z_v = j,  \hzz)$ and $\PP(g_2 |Z_v = j,  \zsv)$ is bounded for any $j \in [k]$. 

\begin{lemma} \label{lemma:approximation}
Suppose the initial estimate $\widehat{\z}^{(0)}$ satisfies $l(\z, \widehat{\z}^{(0)}) \le \eta_n$.
The probability that $G_2$ does not belong to $\G_v$ is at most $n^{-\omega(1)}$, i.e.,
\begin{align}
\PP\left(G_2 \notin \G_v \big|Z_v = i,  \zsv \right) = \PP\bigg(\sum_{a=1}^{D'_v} Y_a > c_n \log n \big|Z_v = i,  \zsv \bigg) \le n^{-\omega(1)}.
\end{align}
Furthermore, there exist $L_h = n^{o(1)}$ and $L_l = n^{-o(1)}$ such that for all $g_2 \in \G_v$ and $j \in [k]$,
\begin{align}
L_l \le \frac{\PP(g_2 |Z_v = j,  \hzz)}{\PP(g_2 |Z_v = j,  \zsv)} \le L_h.
\end{align}
\end{lemma}
\begin{proof} 
See Appendix~\ref{appendix:approximation}.
\end{proof}

We now return to~\eqref{eq:measure}. When $g_2 \in \G_v$ and $\PP(g_2 |Z_v = j, \hzz)\cdot p_j \ge \PP(g_2 |Z_v = i, \hzz)\cdot p_i$, we have $\PP(g_2 |Z_v = i, \zsv) \le (p_jL_u/p_iL_l) \PP(g_2 |Z_v = j, \zsv)$ according to Lemma~\ref{lemma:approximation}. Thus, one can upper-bound~\eqref{eq:measure} as 
\begin{align}
&\sum_{j \ne i} \sum_{g_2 \in \G_v} \PP\left(g_2 \big|Z_v = i,  \zsv \right) \mathbbm{1}\left\{\PP\left(g_2 \big|Z_v = j, \hzz \right) \cdot p_j \ge \PP\left(g_2 \big|Z_v = i, \hzz \right) \cdot p_i \right\} + \sum_{g_2 \notin \G_v}\PP(g_2 |Z_v = i, \zsv) \notag  \\
&\le \sum_{j \ne i} \sum_{g_2 \in \G_v} \frac{p_jL_u}{p_iL_l} \min \left\{ \PP\left(g_2 \big|Z_v = i,  \zsv \right), \PP\left(g_2 \big|Z_v = j,  \zsv \right) \right\} + n^{-\omega(1)} \\
&\le  \sum_{j \ne i} \frac{p_jL_u}{p_iL_l} \sum_{g_2} \min \left\{ \PP\left(g_2 \big|Z_v = i,  \zsv \right), \PP\left(g_2 \big|Z_v = j,  \zsv \right) \right\} + n^{-\omega(1)}. \label{eq:27}
\end{align}
Lemma~\ref{lemma:min} below bounds the term $\sum_{g_2} \min \left\{ \PP\left(g_2 \big|Z_v = i,  \zsv \right), \PP\left(g_2 \big|Z_v = j,  \zsv \right) \right\}$ for any $i \ne j$ in terms of the GCH-divergence $D_+(i,j)$ between communities $\V_i$ and $\V_j$. 

\begin{lemma}\label{lemma:min}
When the ground-truth community vector $\z \in \mathcal{A}_{Z}$ and the sub-hypergraphs $f_1 \in \G^{\z}_{1}$ and $f_2 \in \G^{\z}_{2}$, we have
\begin{align}
\sum_{g_2} \min \left\{ \PP\left(g_2 \big|Z_v = i,  \zsv \right), \PP\left(g_2 \big|Z_v = j,  \zsv \right) \right\} \le n^{-D_+(i,j) + o(1)},
\end{align}
where 
$D_+(i,j) = \max_{t \in [0,1]} \sum_{\m \in \M} t\mu_{\m \oplus i} + (1-t)\mu_{\m \oplus j} - \mu_{\m \oplus i}^t \mu_{\m \oplus j}^{1-t}$.
\end{lemma}

\begin{proof}[Proof of Lemma~\ref{lemma:min}]
As $G_2$ is equivalent to the collection of random variables $\big\{\{X^{\m}_{a} \}_{a=1}^{D_{v,\m}} \big\}_{\m \in \M}$, we have
\begin{align}
&\sum_{g_2} \min \left\{ \PP\left(g_2 \big|Z_v = i,  \zsv \right), \PP\left(g_2 \big|Z_v = j,  \zsv \right) \right\}  \notag \\
&=\sum_{\m \in \M} \sum_{ \{x^{\m}_a\}_{a=1}^{D_{v,\m}} } \min \left\{ \prod_{\m \in \M}\prod_{a=1}^{D_{v,\m}} \PP\left(x_a^{\m} \big|Z_v = i,  \zsv \right), \prod_{\m \in \M}\prod_{a=1}^{D_{v,\m}} \PP\left(x_a^{\m} \big|Z_v = j,  \zsv \right) \right\}  \\
&\le \sum_{\m \in \M} \sum_{ \{x^{\m}_a\}_{a=1}^{D_{v,\m}} } \left(\prod_{\m \in \M}\prod_{a=1}^{D_{v,\m}} \PP\left(x_a^{\m} \big|Z_v = i,  \zsv \right)\right)^t  \left(\prod_{\m \in \M}\prod_{a=1}^{D_{v,\m}} \PP\left(x_a^{\m} \big|Z_v = j,  \zsv \right)\right)^{1-t}  \label{eq:32} \\
&= \prod_{\m \in \M}\prod_{a=1}^{D_{v,\m}} \sum_{x_a^{\m} \in \{0,1\}} \PP\left(x_a^{\m} \big|Z_v = i,  \zsv \right)^t \PP\left(x_a^{\m} \big|Z_v = j,  \zsv \right)^{1-t} \\
&= \prod_{\m \in \M} \left[ \left(\frac{Q_{\m \oplus i} \log n}{n^{d-1}} \right)^t \left(\frac{Q_{\m \oplus j} \log n}{n^{d-1}} \right)^{1-t} + \left(1 - \frac{Q_{\m \oplus i} \log n}{n^{d-1}} \right)^t \left(1 - \frac{Q_{\m \oplus j} \log n}{n^{d-1}} \right)^{1-t} \right]^{D_{v,\m}}, \label{eq:34}
\end{align}	
where~\eqref{eq:32}-\eqref{eq:34} hold for any $t \in [0,1]$.
By applying a Taylor series expansion, we have
\begin{align}
&\left(1 - \frac{Q_{\m \oplus i} \log n}{n^{d-1}} \right)^t \left(1 - \frac{Q_{\m \oplus j} \log n}{n^{d-1}} \right)^{1-t} \notag \\
&= \exp\left\{\log \left\{\frac{(n^{d-1}-Q_{\m \oplus i } \log n)^t (n^{d-1}-Q_{\m \oplus j } \log n)^{1-t}}{n^{d-1}} \right\} \right\} \\
&= \exp\left\{t \log\left(n^{d-1}\left(1 - \frac{Q_{\m \oplus i} \log n}{n^{d-1}}\right) \right) + (1-t) \log\left(n^{d-1}\left(1 - \frac{Q_{\m \oplus j} \log n}{n^{d-1}}\right) \right) - \log(n^{d-1}) \right\} \\
&= \exp\left\{-t \frac{Q_{\m \oplus i}\log n}{n^{d-1}} - (1-t) \frac{Q_{\m \oplus j}\log n}{n^{d-1}} - \frac{1}{2}\left(tQ^2_{\m \oplus i} + (1-t)Q^2_{\m \oplus j} \right) \frac{(\log n)^2}{n^{2d-2}} + \mathcal{O}\left(\frac{(\log n)^3}{n^{3d-3}} \right) \right\} \\
&= 1 - t \frac{Q_{\m \oplus i}\log n}{n^{d-1}} - (1-t) \frac{Q_{\m \oplus j}\log n}{n^{d-1}} - \mathcal{O}\left(\frac{(\log n)^2}{n^{2d-2}} \right), \label{eq:neg}
\end{align}
where in~\eqref{eq:neg} one needs to be aware that the lower order term  has a negative coefficient.
Thus, one can rewrite~\eqref{eq:34} as  
\begin{align}
&\exp\left\{\sum_{\m \in \M} D_{v,\m} \cdot \log \left[1 - t \frac{Q_{\m \oplus i}\log n}{n^{d-1}} - (1-t) \frac{Q_{\m \oplus j}\log n}{n^{d-1}} + \left(\frac{Q_{\m \oplus i} \log n}{n^{d-1}} \right)^t \left(\frac{Q_{\m \oplus j} \log n}{n^{d-1}} \right)^{1-t} -  \mathcal{O}\left(\frac{(\log n)^2}{n^{2d-2}}\right) \right]  \right\} \notag \\
&= \exp\left\{-(\log n)\sum_{\m \in \M} \frac{D_{v,\m}}{n^{d-1}} \cdot \left( t Q_{\m \oplus i} + (1-t) Q_{\m \oplus j} - Q_{\m \oplus i}^t Q_{\m \oplus j}^{1-t} +  \mathcal{O}\left(\frac{\log n}{n^{d-1}}\right) \right)  \right\} \\
&\le  n^{-\sum_{\m \in \M} R'_{\m}\left( t Q_{\m \oplus i} + (1-t) Q_{\m \oplus j} - Q_{\m \oplus i}^t Q_{\m \oplus j}^{1-t}\right) + \mathcal{O}(\gamma_n/\log n) }, \label{eq:42}
\end{align}
where~\eqref{eq:42} is due to the fact that $\left(1 - (3\gamma_n/\log n) \right) R_{\m} \le D_{v,\m} \le (1 + n^{-\frac{1}{2}+\frac{\delta}{2}}) R_{\m}$.
Since~\eqref{eq:42} is valid for any $t \in [0,1]$ and recall that $\mu_{\m \oplus i} = R'_{\m} Q_{\m \oplus i}$, we eventually obtain that 
\begin{align}
\sum_{g_2} \min \left\{ \PP\left(g_2 \big|Z_v = i,  \zsv \right), \PP\left(g_2 \big|Z_v = j,  \zsv \right) \right\} &\le n^{-\max_{t \in [0,1]} \sum_{\m \in \M}  t\mu_{\m \oplus i} + (1-t) \mu_{\m \oplus j} - \mu_{\m \oplus i}^t \mu_{\m \oplus j}^{1-t} + o(1)} \\
&= n^{-D_+(i,j) +o(1)}.
\end{align}
This completes the proof of Lemma~\ref{lemma:min}.
\end{proof}

Since $\min_{i,j\in [k]: i \ne j} D_+(i,j) > 1$, there exists an $\varepsilon > 0$  such that $D_+(i,j) > 1 + \varepsilon$ for all $i\ne j$. By also noting that $L_u/L_l = n^{o(1)}$, we have that $\sum_{j \ne i} (p_j L_u/p_iL_l)\sum_{g_2} \min \{ \PP\left(g_2 |Z_v = i,  \zsv \right), \PP\left(g_2 \big|Z_v = j,  \zsv \right) \} \le n^{-(1+\varepsilon)+o(1)}$. Thus, one can bound the error probability for node $v$ from above as  
\begin{align}
\PP\left(\hZv \ne z_v  \right) \le n^{-(1+\varepsilon)+o(1)} + n^{-\omega(1)}.
\end{align}
Note that the above analysis is also valid for nodes that belong to any other communities (not necessarily  community $\V_i$), thus one can take a union bound over all the $n$ nodes to obtain that 
\begin{align}
\PP\left(\exists v \in [n]: \hZv \ne z_v  \right) = \PP\left(\widehat{\Z} \ne \z \right)  \le n^{-\varepsilon/2}
\end{align} 
when  $n$ is sufficiently large. 
This means that all the nodes can be recovered correctly with probability at least $1 - n^{-\varepsilon/2}$. 

\subsection{The Overall Success Probability}
Let $\EE_{\mathrm{suc}}$ be the event that $l(\z,\widehat{\Z}) = 0$. From the analysis of Stage 2, we know that for all $\z \in \mathcal{A}_{Z}, f_1 \in \G^{\z}_{1}, f_2 \in \G^{\z}_{2}$, and $\widehat{\z}^{(0)}$ satisfying $l(\z,\widehat{\z}) \le  \eta_n$, 
\begin{align}
\PP(\EE_{\mathrm{suc}}|\z,f_1,f_2, \widehat{\z}^{(0)}) \ge 1-n^{-\varepsilon/2}, \label{eq:suc}
\end{align}
Therefore, the overall success probability is 
\begin{align}
\PP\left(l(\Z,\widehat{\Z}) = 0 \right) &= \sum_{\z} \PP(\z) \sum_{f_1, f_2} \PP(f_1, f_2) \sum_{\widehat{\z}^{(0)}} \PP(\widehat{\z}^{(0)}|\z,f_1,f_2)  \cdot \PP(\EE_{\mathrm{suc}}|\z,f_1,f_2, \widehat{\z}^{(0)})\\
&\ge \sum_{\z \in \mathcal{A}_{Z}}\PP(\z) \sum_{f_1 \in \G^{\z}_{1}, f_2 \in \G^{\z}_{2}} \PP(f_1, f_2) \sum_{\widehat{\z}^{(0)}: l(\widehat{\z}^{(0)}, \z) \le \eta_n} \PP(\widehat{\z}^{(0)}|\z,f_1,f_2)  \cdot \PP(\EE_{\mathrm{suc}}|\z,f_1,f_2, \widehat{\z}^{(0)}) \\
& \ge (1 - \exp(-\Theta(n^{\delta}))) \cdot (1-2\sqrt{\epsilon_n}) \cdot  (1-\sqrt{\epsilon_n}) \cdot (1 - n^{-\varepsilon/2}) \label{eq:last} \\
&= 1 - o(1),
\end{align}
where inequality~\eqref{eq:last} follows from Lemma~\ref{lemma:size}, Lemma~\ref{lemma:realization}, the definition of good realization $f_1$ in Definition~\ref{def:subgraph}, and Eqn.~\eqref{eq:suc}.  This means that exact recovery is achievable.

\section{Proof of Converse (Theorem~\ref{thm:converse})} \label{sec:converse}
In this section, we show that when the model parameters $\big(\p, \{Q_{\T}\}_{\T \in \TT}\big)$ satisfy $\min_{i,j\in [k]: i \ne j} D_+(i,j) < 1$, exact recovery is impossible. 
This converse proof is inspired by that for the SBM~\cite{abbe2015community}, but is adapted to the $d$-HSBM setting.

First, we recall from Lemma~\ref{lemma:size} that with high probability the number of nodes in each community $\V_j$ is tightly concentrated around the expectation $np_j$ (i.e., the ground-truth community vector $\Z \in \mathcal{A}_Z$). Hence, we consider a fixed $\z \in \mathcal{A}_Z$ from now on. 
Let $S$ be a random set that contains $n/(\log n)^3$ randomly selected nodes from $[n]$. By applying the Chernoff bound, we can show that  with probability $1 - \exp(-\Theta(n^{\delta}))$ over the selection process, the number of nodes in both $\V_j$ and $S$, denoted by $\V_{j}^S$, satisfies 
\begin{align}
(1-n^{-\frac{1}{2}+\frac{\delta}{2}})^2 \frac{np_j}{(\log n)^3} \le |\V_{j}^S| \le (1+n^{-\frac{1}{2}+\frac{\delta}{2}})^2 \frac{np_j}{(\log n)^3}, \quad \forall j \in [k],\label{eq:S}
\end{align}   
and thus the number of nodes in both $\V_j$ and $S^c$, denoted by $\V_{j}^{S^c}$, satisfies
\begin{align}
np_j \left(1 - \frac{1}{(\log n)^3} -2n^{-\frac{1}{2}+\frac{\delta}{2}}\right) \le |\V_{j}^{S^c}| \le np_j \left(1 - \frac{1}{(\log n)^3} +2n^{-\frac{1}{2}+\frac{\delta}{2}}\right) , \quad \forall j \in [k]\label{eq:Sc}
\end{align}
for sufficiently large $n$. We then consider a fixed set $S$ that satisfies~\eqref{eq:S} and~\eqref{eq:Sc}. Let 
\begin{align}
f_{i,j}(t) \triangleq \sum_{\m \in \M} t\mu_{\m \oplus i} + (1-t)\mu_{\m \oplus j} - \mu_{\m \oplus i}^t \mu_{\m \oplus j}^{1-t},
\end{align}
and note that $D_+(i,j) = \max_{t \in [0,1]} f_{i,j}(t)$. To obtain the maximizer of $f_{i,j}(t)$, we set $f'_{i,j}(t) = 0$ and this implies 
\begin{align}
\sum_{\m \in \M} \mu_{\m \oplus i}^t  \mu_{\m \oplus j}^{1-t} \log\left(\frac{\mu_{\m \oplus i}}{\mu_{\m \oplus j}}\right) = \sum_{\m \in \M} \mu_{\m \oplus i} - \mu_{\m \oplus j}. \label{eq:dao}
\end{align} 
Let $\mc{\tau_{\m}^{i,j}} \triangleq \lfloor \mu_{\m \oplus i}^{t^*} \mu_{\m \oplus j}^{1-t^*} \log n \rfloor$ for each $\m \in \M$, where $t^*$ is set to satisfy~\eqref{eq:dao}.

\begin{definition}
For each node $v \in S$, let 
$N_v^{\m}$ denote the number of hyperedges that contains $v$ and other $d-1$ nodes from $[n] \setminus S$ that have community assignment $\m \in \M$.  A node $v \in S$ is said to be \emph{ambiguous} if $N_v^{\m} = \mc{\tau_{\m}^{i,j}}$ for all $\m \in \M$.  
\end{definition}

In the following, we show that if $D_+(i,j) < 1$, there is at least one ambiguous node in $\V_i^S$ and one ambiguous node in $\V_j^S$. For a node $v \in \V_i^S$, $N_v^{\m}$ equals the sum of $\R_{\m}$ i.i.d. Bernoulli random variables $\{B^{\m}_r \}_{r=1}^{\R_{\m}}$ with expectation $Q_{\m \oplus i}\frac{\log n}{n^{d-1}}$, where $\R_{\m} = \prod_{s\in[k]}\binom{|\V_{s}^{S^c}|}{m_s}$. Due to~\eqref{eq:S} and~\eqref{eq:Sc}, one can show that 
\begin{align}
\left(1 - \frac{1}{(\log n)^3} - 2n^{-\frac{1}{2}+\frac{\delta}{2}} \right) \le \frac{\R_{\m}}{R_{\m}} \le \left(1 - \frac{1}{(\log n)^3} + 2n^{-\frac{1}{2}+\frac{\delta}{2}} \right). \label{eq:rm}
\end{align}
Following~\cite[Exercise 2.2]{csiszar2011information}, one can show that the probability of $N_v^{\m} = \tau_{\m}^{i,j}$ is 
\begin{align}
\PP\left(\sum_{r=1}^{\R_{\m}} B^{\m}_r = \tau_{\m}^{i,j} \right) &= \exp\Bigg\{-\R_{\m} \DD\left(\frac{R_{\m}}{\R_{\m}} \frac{Q^{t^*}_{\m\oplus i}Q^{1-t^*}_{\m\oplus j}\log n}{n^{d-1}} \Bigg\Vert \frac{Q_{\m \oplus i} \log n}{n^{d-1}} \right) - \frac{1}{2}\log(2\pi \R_{\m}) \notag  \\
&\qquad\qquad\qquad- \frac{1}{2}\log\left(1 - \frac{R_{\m}}{\R_{\m}} \frac{Q^{t^*}_{\m\oplus i}Q^{1-t^*}_{\m\oplus j}\log n}{n^{d-1}}\right) - \frac{1}{2} \log\left(\frac{R_{\m}}{\R_{\m}} \frac{Q^{t^*}_{\m\oplus i}Q^{1-t^*}_{\m\oplus j}\log n}{n^{d-1}}\right) -  C_2 \Bigg\}, \label{eq:com1}
\end{align}   
where $C_2 > 0$ is a constant.
By using a Taylor series expansion and the fact that $\R_{\m}/R_{\m}$ is bounded (as shown in~\eqref{eq:rm}), we have 
\begin{align}
&\R_{\m} \DD\left(\frac{R_{\m}}{\R_{\m}} \frac{Q^{t^*}_{\m\oplus i}Q^{1-t^*}_{\m\oplus j}\log n}{n^{d-1}} \Bigg\Vert \frac{Q_{\m \oplus i} \log n}{n^{d-1}} \right) \notag \\
&= R_{\m}\frac{Q^{t^*}_{\m\oplus i}Q^{1-t^*}_{\m\oplus j}\log n}{n^{d-1}} \log\left(\frac{R_{\m}}{\R_{\m}} \left(\frac{Q_{\m \oplus j}}{Q_{\m \oplus i}} \right)^{1-t^*} \right) + \R_{\m}\left(1 - \frac{R_{\m}}{\R_{\m}}\frac{Q^{t^*}_{\m\oplus i}Q^{1-t^*}_{\m\oplus j}\log n}{n^{d-1}}  \right) \notag \\
&\qquad\qquad\qquad\qquad\qquad\qquad \qquad\qquad\qquad\qquad\qquad\qquad \times \log\left(1 + \frac{Q_{m \oplus i}\log n - \frac{R_{\m}}{\R_{\m}}Q^{t^*}_{\m \oplus i} Q^{1-t^*}_{\m \oplus j}\log n }{n^{d-1} - Q_{\m \oplus i}\log n} \right) \\
&= R'_{\m}Q^{t^*}_{\m\oplus i}Q^{1-t^*}_{\m\oplus j}(\log n) (1-t^*) \log\left(\frac{Q_{\m \oplus j}}{Q_{\m \oplus i}} \right) + \frac{\R_{\m}}{R_{\m}} \left( R'_{\m} Q_{\m \oplus i} \log n - \frac{R_{\m}}{\R_{\m}}R'_{\m}  Q^{t^*}_{\m\oplus i}Q^{1-t^*}_{\m\oplus j}\log n + \mathcal{O}(1/(\log n)^2) \right) \\
& = (\log n) \left[\mu^{t^*}_{\m\oplus i}\mu^{1-t^*}_{\m\oplus j} (1-t^*) \log\left(\frac{\mu_{\m \oplus j}}{\mu_{\m \oplus i}} \right) + \mu_{\m \oplus i} - \mu^{t^*}_{\m \oplus i} \mu^{1-t^*}_{\m \oplus j} + \mathcal{O}(1/(\log n)^3) \right]. \label{eq:com2}
\end{align}
Since $\R_{\m} = \Theta(n^{d-1})$, we also have 
\begin{align}
&\frac{1}{2}\log(2\pi \R_{\m}) + \frac{1}{2}\log\left(1 - \frac{R_{\m}}{\R_{\m}} \frac{Q^{t^*}_{\m\oplus i}Q^{1-t^*}_{\m\oplus j}\log n}{n^{d-1}}\right) + \frac{1}{2} \log\left(\frac{R_{\m}}{\R_{\m}} \frac{Q^{t^*}_{\m\oplus i}Q^{1-t^*}_{\m\oplus j}\log n}{n^{d-1}}\right) \\
&\qquad\qquad\qquad\qquad= \frac{d-1}{2}\log n -  \frac{d-1}{2}\log n + \mathcal{O}(\log\log n) = \mathcal{O}(\log\log n). \label{eq:com3}
\end{align}
Combining~\eqref{eq:com1},~\eqref{eq:com2}, and~\eqref{eq:com3}, we then have
\begin{align} 
\PP(N_v^{\m} = \tau_{\m}^{i,j}) = n^{-\left[(1-t^*)\mu^{t^*}_{\m \oplus i}\mu^{1-t^*}_{\m \oplus j}\log(\mu_{\m \oplus j}/\mu_{\m \oplus i}) + \mu_{\m \oplus i} - \mu^{t^*}_{\m \oplus i}\mu^{1-t^*}_{\m \oplus j} + o(1) \right]}. 
\end{align}
Taking all $\m \in \M$ into account, we obtain the probability that a node $v \in \V_i^S$ is ambiguous as follows:
\begin{align}
\PP(v\in \V_i^S \ \text{is ambiguous}) &= \PP(\forall \m \in \M: N_v^{\m} = \tau_{\m}^{i,j}) \\
&= n^{-\left[\sum_{\m \in \M} (1-t^*)\mu^{t^*}_{\m \oplus i}\mu^{1-t^*}_{\m \oplus j}\log(\mu_{\m \oplus j}/\mu_{\m \oplus i}) + \mu_{\m \oplus i} - \mu^{t^*}_{\m \oplus i}\mu^{1-t^*}_{\m \oplus j} + o(1) \right]} \label{eq:in1} \\
&= n^{-[D_+(i,j)+o(1)]}, \label{eq:in2}
\end{align}
where~\eqref{eq:in1} is due to the independence of $\{N_v^{\m} \}_{\m \in \M}$, and~\eqref{eq:in2} holds since $t^*$ satisfies~\eqref{eq:dao}.
Similarly, one can also show that for a node $v \in \V_j^S$, 
\begin{align}
\PP(v\in \V_j^S \ \text{is ambiguous}) &= \PP(\forall \m \in \M: N_v^{\m} = \tau_{\m}^{i,j}) \\
&= n^{-\left[\sum_{\m \in \M} t^*\mu^{t^*}_{\m \oplus i}\mu^{1-t^*}_{\m \oplus j}\log(\mu_{\m \oplus i}/\mu_{\m \oplus j}) + \mu_{\m \oplus j} - \mu^{t^*}_{\m \oplus i}\mu^{1-t^*}_{\m \oplus j} + o(1) \right]} \\
&= n^{-[D_+(i,j)+o(1)]}.
\end{align}
By noting that $D_+(i,j) < 1$ and the cardinalities of both $\V_i^S$ and $\V_j^S$ scale as $\Theta(n/(\log n)^3)$, one can show that with probability $1 - o(1)$, there is at least one ambiguous node in $\V_i^S$ (denoted by $v_1$) and also one ambiguous node in $\V_j^S$ (denoted by $v_2$). 

In addition, we prove that with high probability, node $v_1$ (resp. $v_2$) is not connected to any node in $S$. This is because the number of hyperedges that contains $v_1$ (resp. $v_2$) and another node in $S$ is at most $|S|\binom{n}{d-2} \le n^{d-1}/(\log n)^3$, and the probability of each hyperedge is at most $Q_{\max}(\log n)/n^{d-1}$, thus the probability that $v_1$ (resp. $v_2$) does not have any connection with other nodes in $S$ is at least 
\begin{align}
\left(1 - Q_{\max}\frac{\log n}{n^{d-1}} \right)^{\frac{n^{d-1}}{(\log n)^3}} \ge e^{-\frac{2Q_{\max}}{(\log n)^2}} \ge 1 - \frac{4Q_{\max}}{(\log n)^2},
\end{align}   
for sufficiently large $n$.
Finally, note that both $v_1$ and $v_2$ are not connected to any node in $S$, and both of them are ambiguous (i.e., have the same number of hyperedges $\{N_v^{\m} \}_{\m\in\M}$ outside $S$), thus it is impossible to distinguish them and  to achieve exact recovery.

\section{Conclusion, Discussions, and Future Directions} \label{sec:conclusion}

This paper establishes a sharp phase transition for exact recovery in the general $d$-HSBM, apart from a small subset of generative distributions such that there exists two communities with the same second-order degree profiles. We also develop a polynomial-time algorithm (with theoretical guarantees) that achieves the information-theoretic limit, showing that there is no information-computation gap. Our two-stage algorithm is based on hypergraph spectral clustering and local refinement steps.

Next, we discuss  some connections between our results and related works. 
\begin{enumerate}
\item The second-order degree profile condition for our algorithm to succeed is milder than the conditions of several existing hypergraph spectral clustering methods, e.g.,~\cite{ghoshdastidar2017consistency, chien2019minimax}, which typically require the $k$-th largest singular value of the expected hypergraph Laplacian $\E(\mathbf{L})$ to be sufficiently large (referred to as the \emph{singular value condition} below). Thus, our achievability result (Theorem~2) is applicable to a larger set of parameters. To be specific:
\begin{itemize}
\item[--] When the second-order degree profile condition is violated, there must exist two communities having the same second-order degree profile, which implies that the columns corresponding to these two communities in $\E(\mathbf{L})$ are the same. Thus,  the rank of $\E(\mathbf{L})$ is less than $k$ and the $k$-th largest singular value equals zero. This means that the singular value condition is also violated.

\item[--] When the singular value condition is violated, it does not necessarily imply that the second-order degree profile condition is violated. For example, suppose node $u \in \mathcal{V}_1$, node $v \in \mathcal{V}_2$, and their corresponding columns in $\E(\mathbf{L})$ satisfy $\E(\mathbf{L})_u = 2\E(\mathbf{L})_v$, then the rank of $\E(\mathbf{L})$ is less than $k$ and the $k$-th largest singular value is zero (i.e., the singular value condition is violated). However, since the columns corresponding to $\mathcal{V}_1$ and $\mathcal{V}_2$ are different (though they are linearly dependent), the second-order degree profiles of $\mathcal{V}_1$ and $\mathcal{V}_2$ are different, which does not imply that the second-order degree profile condition is violated.
\end{itemize}

\item Another work that is closely related to ours is~\cite{yuan2022testing}, which considered the fundamental question of whether communities exist or not in a hypergraph. They characterized the condition under which a hypergraph generated according to a $d$-HSBM can be successfully distinguished from a hypergraph generated according to an Erd\H{o}s–R\'{e}nyi hypergraph model, where the $d$-HSBM contains $k$ equal-sized communities and the hyperedge probabilities are assumed to be either $a_n$ or $b_n$ (depending on whether all $d$ nodes belong to a same community). Their main messages are that (i) when $a_n, b_n \in o(n^{-d+1})$, these two models are indistinguishable; (ii) when $a_n, b_n \in \omega(n^{-d+1})$, their proposed test ensures the two models to be distinguishable with probability approaching one; (iii)  when $a_n, b_n \in \Theta(n^{-d+1})$, the two models are distinguishable if the so-called SNR is greater than a certain threshold, while indistinguishable if the SNR is below another threshold. Comparing~\cite{yuan2022testing} with our work, it is interesting to note that the phase transition occurs in the \emph{constant average degrees regime} for detecting the existence of communities~\cite{yuan2022testing}, while the phase transition occurs in the \emph{logarithmic average degrees regime} for exactly recovering communities.

\item Community detection in hypergraphs is also related to the planted $k$-SAT problem~\cite{feldman2018complexity}, in which the objective is to identify a planted assignment $\sigma \in \{\pm 1\}^n$ of $n$ Boolean variables $\{x_1,x_2 \ldots, x_n\}$ given a sequence of randomly generated \emph{$k$-clauses}, where each $k$-clause is a collection of $k$ distinct elements chosen from $\{x_1,x_2 \ldots, x_n\}$ and their negations $\{\bar{x}_1,\bar{x}_2 \ldots, \bar{x}_n\}$. Let $\mathcal{X}_k$ be the set of all $k$-clauses (with $|\mathcal{X}_k| = \binom{2n}{k}$), and $Q: \{\pm 1\}^k \to [0,1]$ be a probability distribution\footnote{Probability distributions of special interests are that satisfy $Q(\{-1,\ldots, -1\}) = 0$, which correspond to a common assumption that only satisfied clauses are allowed to appear.} on $\{\pm 1\}^k$ such that $\sum_{\mathbf{y} \in \{\pm 1\}^k}Q(\mathbf{y}) = 1$.   At each time when we generate a $k$-clause, the probability of a $k$-clause $\mathbf{c} = [c_1, \ldots, c_k]$ (where $c_i \in \{x_1,x_2 \ldots, x_n\} \cup \{\bar{x}_1,\bar{x}_2 \ldots, \bar{x}_n\}$) being selected is 
$\PP(\mathbf{c} \mbox{ is selected}) = \frac{Q(\sigma(\mathbf{c}))}{\sum_{\mathbf{c}' \in \{\pm 1\}^k} Q(\sigma(\mathbf{c}')) },$
where $\sigma(\mathbf{c}) \in \{\pm 1\}^k$ is the assignment of the $k$ elements in $\mathbf{c}$ under the assignment $\sigma$. In the planted $k$-SAT problem, we generate $M$ independent $k$-clauses, and the question of interest is to find how many clauses $M$ are required for successful recovery of the assignment $\sigma$ with high probability.
This planted $k$-SAT problem can be viewed as a random hypergraph $(\mathcal{V}, \mathcal{E})$, where the node set $\mathcal{V} = \{x_1,x_2 \ldots, x_n\} \cup \{\bar{x}_1,\bar{x}_2 \ldots, \bar{x}_n\}$ is of size $2n$, and the edge set $\mathcal{E}$ contains $M$ $k$-uniform hyperedges with each one corresponding to a randomly generated $k$-clause. The nodes in $\mathcal{V}$ are partitioned into two communities that  correspond to `$+1$' and `$-1$', where the two communities are of exactly equal sizes by construction. While  the planted $k$-SAT problem can be approximately viewed as the HSBM problem studied in this work, there are also several notable differences. First, the generation process of $k$-clauses is different from the generation process of hyperedges in the HSBM---the former allows each $k$-clause to be selected for multiple times, while the latter only allows each $k$-uniform hyperedge to be selected once. Second, the assignments of nodes in the planted $k$-SAT problem are strongly correlated, e.g., the signs of $x_i$ and $\bar{x}_i$ must be different, while there is no such restriction in the HSBM.

$\quad$ Despite the differences, our algorithm is applicable to the planted $k$-SAT problem. One can first convert the $k$-SAT problem to a hypergraph with $2n$ nodes and $M$ $k$-uniform hyperedges, construct the corresponding trimmed hypergraph Laplacian, and then
 apply our spectral clustering method (lines 2-16 in Algorithm~1) to obtain an initial assignment $\widehat{\sigma}^{(0)}$ of nodes $\{x_1,x_2 \ldots, x_n\} \cup \{\bar{x}_1,\bar{x}_2 \ldots, \bar{x}_n\}$. It is expected that this stage leads to an almost exact recovery of the true assignment $\sigma$, as long as $M \in \omega(n)$ (corresponding to hyperedge probabilities being $\omega(1/n^{k-1})$ in the HSBM). In the second stage, one can use the local MAP estimation for each of the $2n$ nodes (lines 17-20 in Algorithm~1) to refine the assignments. However, since $x_i$ and $\bar{x}_i$ in the planted $k$-SAT problem are of different signs, one can instead choose to \emph{refine each pair $(x_i,\bar{x}_i)$ jointly} via the local MAP estimation, which may lead to a better performance. It is expected that when the probability distribution $Q: \{\pm 1\}^k \to [0,1]$ in the planted $k$-SAT problem is specialized to a simple function whose values depend only on the number of `$+1$' in the input (in which case the distribution $Q$ is equivalent to the hyperedge  probabilities $\{Q_{\mathbf{T}}\}_{\mathbf{T} \in \mathcal{T}}$ in the HSBM), the second stage leads to exact recovery of the true assignment with high probability if $M \in \Theta(n \log n)$ (corresponding to the logarithmic average degrees regime in the HSBM). We also expect that the GCH-divergence plays a role in the minimum pre-constant of $\Theta(n \log n)$; however, this pre-constant may not be obtained as a direct consequence of our result due to the several important differences between the planted $k$-SAT problem and HSBM.
\end{enumerate}

Finally, we put forth two promising directions for future work.

\begin{enumerate}
	\item Our algorithm fails if the parameters belong to $\Xi$ because we apply the hypergraph spectral clustering method to the processed hypergraph Laplacian $\AH$ (rather than the observed adjacency tensor $\A$). This pre-processing step from $\A$ to $\AH$ annihilates some salient information for distinguishing two communities with the same second-order degree profile. Thus, any clustering algorithms that rely merely on $\AH$ must be restricted to this second-order degree profile condition. On the other hand, we conjecture that the second-order degree profile condition is \emph{not necessary}, and this issue may be circumvented if one directly applies clustering algorithms to the adjacency tensor $\A$ (such as the tensor-based method proposed in~\cite{ke2019community}). As shown empirically in~\cite[Section~3.4]{ke2019community} (particularly in Figure~3), their method avoids unwanted information loss caused by projecting hypergraphs to weighted graphs under a variety of parameter settings. \mc{Unfortunately, their concentration tools for random tensors are only applicable when the average degree is $\omega(\log^2(n))$, and thus are not powerful enough for the logarithmic average degrees regime considered in this paper. The analysis in~\cite{ke2019community} of tensor concentration relies on the notion of the \emph{incoherent tensor operator norm}, the properties of the tensor, as well as concentration inequalities such as the Bernstein's inequality and Chernoff bound. In contrast, the concentration of random matrices is relatively well understood, and the analysis in this paper relies mainly on techniques from random matrix theory.
 In future work, it is interesting to investigate whether tensor-based methods can be applied to hypergraphs with logarithmic average degrees, and to validate whether our conjecture that the exact recovery threshold $\min_{i,j\in [k]: i \ne j} D_+(i,j) =1$ holds even without the condition on the second-order degree profile discussed in Section~\ref{sec:distance}. }

\item It would also be interesting to extend our theory to even more general settings and other variants of the HSBM, such as the non-uniform HSBM (as proposed in~\cite{ghoshdastidar2015provable,ghoshdastidar2015spectral,ghoshdastidar2017consistency}), HSBM with overlapping communities, weighted or labelled HSBMs, HSBM with side information, etc. 
\end{enumerate}

\appendices

\section{Proof of Lemma~\ref{lemma:M}} \label{appendix:M}
Without loss of generality, we assume nodes $u$ and $v$ respectively belong to communities $\V_i$ and $\V_j$. Since we require $\big(\p, \{Q_{\T}\}_{\T \in \TT}\big) \notin \Xi$, there must exist a $s \in [k]$ such that
\begin{align}
\sum_{\m \in \M: m_s \ge 1} m_s \mu_{\m \oplus i}  \ne \sum_{\m \in \M: m_s \ge 1} m_s \mu_{\m \oplus j}.
\label{eq:assume}
\end{align}
Let $R_{\m,s} \triangleq \binom{np_s}{m_s-1} \cdot \prod_{a\in[k]\setminus\{s\}}\binom{np_a}{m_a}.$ 
Thus, we have 
\begin{align}
\Vert \MM_u - \MM_v \Vert_2^2  &\ge \sum_{w \in \V_s} (\mathrm{M}_{w,u} - \mathrm{M}_{w,v})^2 \\
&= \sum_{w \in \V_s} \left(\sum_{\m\in\M:m_s \ge 1} R_{\m,s}Q_{\m \oplus i} \frac{\gamma_n}{n^{d-1}} - \sum_{\m\in\M:m_s \ge 1} R_{\m,s}Q_{\m \oplus j} \frac{\gamma_n}{n^{d-1}}\right)^2 \\
&= \sum_{w \in \V_s} \left(\frac{\gamma_n}{np_s} \left[\sum_{\m\in\M:m_s \ge 1} m_s\mu_{\m \oplus i}  - m_s \mu_{\m \oplus j}\right] \right)^2 \label{eq:M1}\\
&= \Omega(\gamma_n^2/n), \label{eq:M2}
\end{align}
where~\eqref{eq:M1} holds since $R_{\m,s} = m_sR_{\m}/(np_s)$ and $R_{\m}Q_{\m \oplus i}/n^{d-1} = \mu_{\m\oplus i}$, and~\eqref{eq:M2} follows from~\eqref{eq:assume} and $|\V_s| \approx np_s$.

\section{Proof of Lemma~\ref{lemma:realization}} \label{appendix:realization}

For any realization $F_1 = f_1$, let $P_{\mathrm{suc}}(f_1,\z)$ be the probability that running a hypergraph spectral clustering method on $G_1$ (which depends on $f_1$) ensures $l(\z, \widehat{\Z}^{(0)}) \le \eta_n$.  From Theorem~\ref{thm:weak}, we have 
\begin{align}
\sum_{f_1} \PP(F_1 = f_1)P_{\mathrm{suc}}(f_1,\z) \ge 1 - \epsilon_n. \label{eq:con}
\end{align}
We now prove Lemma~\ref{lemma:realization} by contradiction. Suppose the probability that $F_1 \in \G^{\z}_{1}$ is less than $1 - \sqrt{\epsilon_n}$, then we have 
\begin{align}
\sum_{f_1} \PP(F_1 = f_1)P_{\mathrm{suc}}(f_1,\z) &< \sum_{f_1 \in \G^{\z}_{1}} \PP(F_1 = f_1) + \sum_{f_1 \notin \G^{\z}_{1}} \PP(F_1 = f_1) (1-\sqrt{\epsilon_n}) \label{eq:def}\\
&=  \sum_{f_1 \in \G^{\z}_{1}} \PP(F_1 = f_1) + (1-\sqrt{\epsilon_n}) \Bigg(1 -  \sum_{f_1 \notin \G^{\z}_{1}} \PP(F_1 = f_1)\Bigg) \\
&< 1 - \sqrt{\epsilon_n} + (1 - \sqrt{\epsilon_n})\cdot \sqrt{\epsilon_n}  \label{eq:ass}\\
&= 1 - \epsilon_n, \label{eq:end}
\end{align}
where~\eqref{eq:def} follows from the fact that $P_{\mathrm{suc}}(f_1,\z) < 1- \sqrt{\epsilon_n}$ for $f_1 \notin \G^{\z}_{1}$ (see Definition~\ref{def:subgraph}), and~\eqref{eq:ass} is due to our assumption. Since Eqns.~\eqref{eq:def}-\eqref{eq:end} contradict with the fact in~\eqref{eq:con}, we obtain that $\PP(F_1 \in \G^{\z}_{1}) \ge 1 - \sqrt{\epsilon_n}$.

Let $N_{\m} \triangleq \prod_{s=1}^k \binom{|\V_s|}{m_s}$ for each $\m \in \M$. For each node $v \in [n]$, the expected number of hyperedges in $F_2$ that contain node $v$ and other $d-1$ nodes with community assignment $\m$ is $\E(D_{v,\m}) = N_{\m}\left(1 - \frac{\gamma_n}{\log n} \right)$. By applying the Chernoff bound, we have 
\begin{align}
\PP\left(D_{v,\m} \le N_{\m}\left(1 - \frac{2\gamma_n}{\log n} \right)  \right) \le \PP \left(D_{v,\m} \le \left(1 - \frac{\gamma_n}{\log n} \right) \E(D_{v,\m}) \right) &\le \exp\left(-\frac{1}{3} \frac{\gamma_n^2}{(\log n)^2}\E(D_{v,\m}) \right) \\
&= \exp\left(-\Theta(n^{d-1} \gamma_n^2/ (\log n)^2 ) \right).
\end{align} 
Taking a union bound over all $\m \in \M$ and all the $n$ nodes, we have that with probability at least $1 - \exp\left(-\Theta(n^{d-1} \gamma_n^2/ (\log n)^2 ) \right)$, every node $v \in [n]$ satisfies 
\begin{align}
\left(1 - (2\gamma_n/\log n) \right) N_{\m} \le D_{v,\m} \le N_{\m}, \quad \forall \m \in \M.
\end{align} 
Combining the fact that $(1 - n^{-\frac{1}{2}+\frac{\delta}{2}}) R_{\m} \le N_{\m} \le (1 + n^{-\frac{1}{2}+\frac{\delta}{2}}) R_{\m}$ (since $\z \in \mathcal{A}_{Z}$), we have 
\begin{align}
\left(1 - (3\gamma_n/\log n) \right) R_{\m} \le D_{v,\m} \le (1 + n^{-\frac{1}{2}+\frac{\delta}{2}}) R_{\m}, \quad \forall \m \in \M.
\end{align}
Thus, $\PP(F_2 \in \G^{\z}_{2}) \ge 1 - \exp\left(-\Theta(n^{d-1} \gamma_n^2/ (\log n)^2 ) \right)$. This completes the proof.

\section{Proof of Lemma~\ref{lemma:approximation}} \label{appendix:approximation}
Note that $\E(\sum_{a=1}^{D'_v} Y_a) \le D'_v \cdot Q_{\max}(\log n)/n^{d-1} = \mathcal{O}(\eta_n \log n)$, since the success probability of each Bernoulli random variable is at most $Q_{\max}(\log n)/n^{d-1}$. In the following, we show that the probability that $\sum_{a=1}^{D'_v} Y_a \ge c_n \log n$ is at most $n^{-\omega(1)}$, where $c_n \log n = \omega(\eta_n \log n)$. Note that 
\begin{align}
\PP\left(\sum_{a=1}^{D'_v} Y_a \ge c_n \log n \right) &= \sum_{\theta=c_n \log n}^{D'_v} \PP\left(\sum_{a=1}^{D'_v} Y_a = \theta \right)\\
&= \sum_{\theta=c_n \log n}^{D'_v}  \binom{D'_v}{\theta} \left(\frac{Q_{\max}\log n}{n^{d-1}} \right)^{\theta} \left(1 - \frac{Q_{\max}\log n}{n^{d-1}} \right)^{D'_v-\theta} \label{eq:28}\\
&\le D'_v \cdot \left(\frac{eD'_v}{c_n \log n} \right)^{c_n \log n} \left(\frac{Q_{\max}\log n}{n^{d-1}} \right)^{c_n \log n} \label{eq:29} \\
&= D'_v \cdot n^{-c_n \log(c_n /\eta_n)}  = n^{-\omega(1)},
\end{align}
where~\eqref{eq:29} follows from the facts that $\binom{n}{k} \le (en/k)^k$ and $\theta = c_n \log n$ maximizes the terms in~\eqref{eq:28}.     

We then prove the second part.
For random variables $\{Y_a\}_{a =1}^{D'_v}$, we have 
\begin{align}
\frac{Q_{\min}}{Q_{\max}} \le &\frac{\PP(Y_a = 1 |Z_v = j, \hzz)}{\PP(Y_a= 1 |Z_v = j, \zsv)} \le \frac{Q_{\max}}{Q_{\min}}, \quad \text{and} \\
1 - \frac{(Q_{\max} - Q_{\min}) \log n}{n^{d-1}} \le &\frac{\PP(Y_a = 0 |Z_v = j, \hzz)}{\PP(Y_a = 0 |Z_v = j, \zsv)} \le 1 + \frac{(Q_{\max} - Q_{\min}) \log n}{n^{d-1}}.
\end{align}
Since $\sum_{a=1}^{D'_v} y_a \le c_n \log n$ for $g_2 \in \G_v$, we have 
\begin{align}
&\frac{\PP(g_2 |Z_v = j,  \hzz)}{\PP(g_2 |Z_v = j,  \zsv)} = \frac{\prod_{a=1}^{D'_v}\PP(y_a |Z_v = j,  \hzz)}{\prod_{a=1}^{D'_v} \PP(y_a|Z_v = j,  \zsv)} \le \left(\frac{Q_{\max}}{Q_{\min}}\right)^{c_n \log n} \left(1 + \frac{(Q_{\max} - Q_{\min}) \log n}{n^{d-1}}\right)^{D'_v - c_n \log n} \triangleq L_h, \\
& \frac{\PP(g_2 |Z_v = j,  \hzz)}{\PP(g_2 |Z_v = j,  \zsv)} = \frac{\prod_{a=1}^{D'_v} \PP(y_a |Z_v = j,  \hzz)}{\prod_{a=1}^{D'_v}\PP(y_a|Z_v = j,  \zsv)} \ge \left(\frac{Q_{\min}}{Q_{\max}}\right)^{c_n \log n} \left(1 - \frac{(Q_{\max} - Q_{\min}) \log n}{n^{d-1}}\right)^{D'_v - c_n \log n} \triangleq L_l,
\end{align}
and note that $L_h = n^{o(1)}$ and $L_l = n^{-o(1)}$.

\ifCLASSOPTIONcaptionsoff
  \newpage
\fi

\bibliographystyle{IEEEtran}
\bibliography{reference}

\begin{thebibliography}{10}
\providecommand{\url}[1]{#1}
\csname url@samestyle\endcsname
\providecommand{\newblock}{\relax}
\providecommand{\bibinfo}[2]{#2}
\providecommand{\BIBentrySTDinterwordspacing}{\spaceskip=0pt\relax}
\providecommand{\BIBentryALTinterwordstretchfactor}{4}
\providecommand{\BIBentryALTinterwordspacing}{\spaceskip=\fontdimen2\font plus
\BIBentryALTinterwordstretchfactor\fontdimen3\font minus
  \fontdimen4\font\relax}
\providecommand{\BIBforeignlanguage}[2]{{%
\expandafter\ifx\csname l@#1\endcsname\relax
\typeout{** WARNING: IEEEtran.bst: No hyphenation pattern has been}%
\typeout{** loaded for the language `#1'. Using the pattern for}%
\typeout{** the default language instead.}%
\else
\language=\csname l@#1\endcsname
\fi
#2}}
\providecommand{\BIBdecl}{\relax}
\BIBdecl

\bibitem{holland1983stochastic}
P.~W. Holland, K.~B. Laskey, and S.~Leinhardt, ``Stochastic blockmodels: First
  steps,'' \emph{Social networks}, vol.~5, no.~2, pp. 109--137, 1983.

\bibitem{abbe2015exact}
E.~Abbe, A.~S. Bandeira, and G.~Hall, ``Exact recovery in the stochastic block
  model,'' \emph{IEEE Transactions on Information Theory}, vol.~62, no.~1, pp.
  471--487, 2015.

\bibitem{mossel2015consistency}
E.~Mossel, J.~Neeman, and A.~Sly, ``Consistency thresholds for the planted
  bisection model,'' in \emph{Proceedings of the forty-seventh annual ACM
  symposium on Theory of computing}, 2015, pp. 69--75.

\bibitem{abbe2015community}
E.~Abbe and C.~Sandon, ``Community detection in general stochastic block
  models: Fundamental limits and efficient algorithms for recovery,'' in
  \emph{IEEE 56th Annual Symposium on Foundations of Computer Science}, 2015,
  pp. 670--688.

\bibitem{abbe2015recovering}
------, ``Recovering communities in the general stochastic block model without
  knowing the parameters,'' \emph{arXiv preprint arXiv:1506.03729}, 2015.

\bibitem{hajek2017information}
B.~Hajek, Y.~Wu, and J.~Xu, ``Information limits for recovering a hidden
  community,'' \emph{IEEE Transactions on Information Theory}, vol.~63, no.~8,
  pp. 4729--4745, 2017.

\bibitem{reeves2019geometry}
G.~Reeves, V.~Mayya, and A.~Volfovsky, ``The geometry of community detection
  via the {MMSE} matrix,'' in \emph{2019 IEEE International Symposium on
  Information Theory (ISIT)}, 2019, pp. 400--404.

\bibitem{jog2015information}
V.~Jog and P.-L. Loh, ``Information-theoretic bounds for exact recovery in
  weighted stochastic block models using the {R}enyi divergence,'' \emph{arXiv
  preprint arXiv:1509.06418}, 2015.

\bibitem{yun2016optimal}
S.-Y. Yun and A.~Proutiere, ``Optimal cluster recovery in the labeled
  stochastic block model,'' \emph{Advances in Neural Information Processing
  Systems}, vol.~29, pp. 965--973, 2016.

\bibitem{yun2014accurate}
------, ``Accurate community detection in the stochastic block model via
  spectral algorithms,'' \emph{arXiv preprint arXiv:1412.7335}, 2014.

\bibitem{chin2015stochastic}
P.~Chin, A.~Rao, and V.~Vu, ``Stochastic block model and community detection in
  sparse graphs: A spectral algorithm with optimal rate of recovery,'' in
  \emph{Conference on Learning Theory}, 2015, pp. 391--423.

\bibitem{hajek2016achieving}
B.~Hajek, Y.~Wu, and J.~Xu, ``Achieving exact cluster recovery threshold via
  semidefinite programming,'' \emph{IEEE Transactions on Information Theory},
  vol.~62, no.~5, pp. 2788--2797, 2016.

\bibitem{hajek2016achieving2}
------, ``Achieving exact cluster recovery threshold via semidefinite
  programming: Extensions,'' \emph{IEEE Transactions on Information Theory},
  vol.~62, no.~10, pp. 5918--5937, 2016.

\bibitem{montanari2016}
A.~Montanari and S.~Sen, ``Semidefinite programs on sparse random graphs and
  their application to community detection,'' in \emph{Proceedings of the
  forty-eighth annual ACM symposium on Theory of Computing}, 2016, pp.
  814--827.

\bibitem{caltagirone2017recovering}
F.~Caltagirone, M.~Lelarge, and L.~Miolane, ``Recovering asymmetric communities
  in the stochastic block model,'' \emph{IEEE Transactions on Network Science
  and Engineering}, vol.~5, no.~3, pp. 237--246, 2017.

\bibitem{agarwal2017multisection}
N.~Agarwal, A.~S. Bandeira, K.~Koiliaris, and A.~Kolla, ``Multisection in the
  stochastic block model using semidefinite programming,'' in \emph{Compressed
  Sensing and its Applications}, 2017, pp. 125--162.

\bibitem{perry2017semidefinite}
A.~Perry and A.~S. Wein, ``A semidefinite program for unbalanced multisection
  in the stochastic block model,'' in \emph{2017 International Conference on
  Sampling Theory and Applications (SampTA)}, 2017, pp. 64--67.

\bibitem{asadi2017compressing}
A.~R. Asadi, E.~Abbe, and S.~Verd{\'u}, ``Compressing data on graphs with
  clusters,'' in \emph{IEEE Int. Symp. Inf. Theory (ISIT)}, 2017, pp.
  1583--1587.

\bibitem{saad2018community}
H.~Saad and A.~Nosratinia, ``Community detection with side information: Exact
  recovery under the stochastic block model,'' \emph{IEEE Journal of Selected
  Topics in Signal Processing}, vol.~12, no.~5, pp. 944--958, 2018.

\bibitem{saad2020recovering}
------, ``Recovering a single community with side information,'' \emph{IEEE
  Transactions on Information Theory}, vol.~66, no.~12, pp. 7939--7966, 2020.

\bibitem{saad2018exact}
------, ``Exact recovery in community detection with continuous-valued side
  information,'' \emph{IEEE Signal Processing Letters}, vol.~26, no.~2, pp.
  332--336, 2018.

\bibitem{mayya2019mutual}
V.~Mayya and G.~Reeves, ``Mutual information in community detection with
  covariate information and correlated networks,'' in \emph{57th Annual
  Allerton Conference on Communication, Control, and Computing (Allerton)},
  2019, pp. 602--607.

\bibitem{abbe2017community}
E.~Abbe, ``Community detection and stochastic block models: recent
  developments,'' \emph{The Journal of Machine Learning Research}, vol.~18,
  no.~1, pp. 6446--6531, 2017.

\bibitem{chertok2010efficient}
M.~Chertok and Y.~Keller, ``Efficient high order matching,'' \emph{IEEE
  Transactions on Pattern Analysis and Machine Intelligence}, vol.~32, no.~12,
  pp. 2205--2215, 2010.

\bibitem{liu2010robust}
H.~Liu, L.~Latecki, and S.~Yan, ``Robust clustering as ensembles of affinity
  relations,'' \emph{Advances in neural information processing systems},
  vol.~23, pp. 1414--1422, 2010.

\bibitem{ghoshdastidar2014consistency}
D.~Ghoshdastidar and A.~Dukkipati, ``Consistency of spectral partitioning of
  uniform hypergraphs under planted partition model,'' \emph{Advances in Neural
  Information Processing Systems}, vol.~27, pp. 397--405, 2014.

\bibitem{ghoshdastidar2015provable}
------, ``A provable generalized tensor spectral method for uniform hypergraph
  partitioning,'' in \emph{International Conference on Machine Learning}, 2015,
  pp. 400--409.

\bibitem{ghoshdastidar2015spectral}
------, ``Spectral clustering using multilinear svd: Analysis, approximations
  and applications,'' in \emph{Proceedings of the AAAI Conference on Artificial
  Intelligence}, vol.~29, no.~1, 2015.

\bibitem{ghoshdastidar2017consistency}
------, ``Consistency of spectral hypergraph partitioning under planted
  partition model,'' \emph{Annals of Statistics}, vol.~45, no.~1, pp. 289--315,
  2017.

\bibitem{pal2019community}
S.~Pal and Y.~Zhu, ``Community detection in the sparse hypergraph stochastic
  block model,'' \emph{Random Structures \& Algorithms}, vol.~59, no.~3, pp.
  407--463, 2021.

\bibitem{cole2020exact}
S.~Cole and Y.~Zhu, ``Exact recovery in the hypergraph stochastic block model:
  A spectral algorithm,'' \emph{Linear Algebra and its Applications}, vol. 593,
  pp. 45--73, 2020.

\bibitem{ahn2017information}
K.~Ahn, K.~Lee, and C.~Suh, ``Information-theoretic limits of subspace
  clustering,'' in \emph{2017 IEEE International Symposium on Information
  Theory (ISIT)}, 2017, pp. 2473--2477.

\bibitem{ahn2018hypergraph}
------, ``Hypergraph spectral clustering in the weighted stochastic block
  model,'' \emph{IEEE Journal of Selected Topics in Signal Processing},
  vol.~12, no.~5, pp. 959--974, 2018.

\bibitem{kim2018stochastic}
C.~Kim, A.~S. Bandeira, and M.~X. Goemans, ``Stochastic block model for
  hypergraphs: Statistical limits and a semidefinite programming approach,''
  \emph{arXiv preprint arXiv:1807.02884}, 2018.

\bibitem{lee2020robust}
J.~Lee, D.~Kim, and H.~W. Chung, ``Robust hypergraph clustering via convex
  relaxation of truncated mle,'' \emph{IEEE Journal on Selected Areas in
  Information Theory}, 2020.

\bibitem{kim2017community}
C.~Kim, A.~S. Bandeira, and M.~X. Goemans, ``Community detection in
  hypergraphs, spiked tensor models, and sum-of-squares,'' in \emph{2017
  International Conference on Sampling Theory and Applications (SampTA)}, 2017,
  pp. 124--128.

\bibitem{ke2019community}
Z.~T. Ke, F.~Shi, and D.~Xia, ``Community detection for hypergraph networks via
  regularized tensor power iteration,'' \emph{arXiv preprint arXiv:1909.06503},
  2019.

\bibitem{angelini2015spectral}
M.~C. Angelini, F.~Caltagirone, F.~Krzakala, and L.~Zdeborov{\'a}, ``Spectral
  detection on sparse hypergraphs,'' in \emph{2015 53rd Annual Allerton
  Conference on Communication, Control, and Computing (Allerton)}, 2015, pp.
  66--73.

\bibitem{lesieur2017statistical}
T.~Lesieur, L.~Miolane, M.~Lelarge, F.~Krzakala, and L.~Zdeborov{\'a},
  ``Statistical and computational phase transitions in spiked tensor
  estimation,'' in \emph{2017 IEEE International Symposium on Information
  Theory (ISIT)}, 2017, pp. 511--515.

\bibitem{chien2019minimax}
I.~E. Chien, C.-Y. Lin, and I.-H. Wang, ``On the minimax misclassification
  ratio of hypergraph community detection,'' \emph{IEEE Transactions on
  Information Theory}, vol.~65, no.~12, pp. 8095--8118, 2019.

\bibitem{lin2017fundamental}
C.-Y. Lin, I.~E. Chien, and I.-H. Wang, ``On the fundamental statistical limit
  of community detection in random hypergraphs,'' in \emph{2017 IEEE
  International Symposium on Information Theory (ISIT)}, 2017, pp. 2178--2182.

\bibitem{chien2018community}
I.~Chien, C.-Y. Lin, and I.-H. Wang, ``Community detection in hypergraphs:
  Optimal statistical limit and efficient algorithms,'' in \emph{International
  Conference on Artificial Intelligence and Statistics}, 2018, pp. 871--879.

\bibitem{ahn2019community}
K.~Ahn, K.~Lee, and C.~Suh, ``Community recovery in hypergraphs,'' \emph{IEEE
  Transactions on Information Theory}, vol.~65, no.~10, pp. 6561--6579, 2019.

\bibitem{liang2021information}
J.~Liang, C.~Ke, and J.~Honorio, ``Information theoretic limits of exact
  recovery in sub-hypergraph models for community detection,'' in \emph{2021
  IEEE International Symposium on Information Theory (ISIT)}, 2021, pp.
  2578--2583.

\bibitem{yuan2021information}
M.~Yuan and Z.~Shang, ``Information limits for detecting a subhypergraph,''
  \emph{Stat}, vol.~10, no.~1, p. e407, 2021.

\bibitem{cover2012elements}
T.~M. Cover and J.~A. Thomas, \emph{Elements of Information Theory}.\hskip 1em
  plus 0.5em minus 0.4em\relax John Wiley \& Sons, 2012.

\bibitem{zhang2016minimax}
A.~Y. Zhang and H.~H. Zhou, ``Minimax rates of community detection in
  stochastic block models,'' \emph{The Annals of Statistics}, vol.~44, no.~5,
  pp. 2252--2280, 2016.

\bibitem{gao2017achieving}
C.~Gao, Z.~Ma, A.~Y. Zhang, and H.~H. Zhou, ``Achieving optimal
  misclassification proportion in stochastic block models,'' \emph{The Journal
  of Machine Learning Research}, vol.~18, no.~1, pp. 1980--2024, 2017.

\bibitem{jain2013low}
P.~Jain, P.~Netrapalli, and S.~Sanghavi, ``Low-rank matrix completion using
  alternating minimization,'' in \emph{STOC}, 2013, pp. 665--674.

\bibitem{keshavan2010matrix}
R.~H. Keshavan, A.~Montanari, and S.~Oh, ``Matrix completion from a few
  entries,'' \emph{IEEE Trans. Inf. Theory}, vol.~56, no.~6, pp. 2980--2998,
  2010.

\bibitem{zhang2020mc2g}
Q.~Zhang, G.~Suh, C.~Suh, and V.~Y.~F. Tan, ``{MC2G}: An efficient algorithm
  for matrix completion with social and item similarity graphs,'' \emph{IEEE
  Transactions on Signal Processing}, vol.~70, pp. 2681--2697, 2022.

\bibitem{vu2014simple}
V.~Vu, ``A simple {SVD} algorithm for finding hidden partitions,'' \emph{arXiv
  preprint arXiv:1404.3918}, 2014.

\bibitem{eckart1936approximation}
C.~Eckart and G.~Young, ``The approximation of one matrix by another of lower
  rank,'' \emph{Psychometrika}, vol.~1, no.~3, pp. 211--218, 1936.

\bibitem{weyl1912asymptotische}
H.~Weyl, ``Das asymptotische verteilungsgesetz der eigenwerte linearer
  partieller differentialgleichungen (mit einer anwendung auf die theorie der
  hohlraumstrahlung),'' \emph{Mathematische Annalen}, vol.~71, no.~4, pp.
  441--479, 1912.

\bibitem{csiszar2011information}
I.~Csisz{\'a}r and J.~K{\"o}rner, \emph{Information Theory: Coding Theorems For
  Discrete Memoryless Systems}.\hskip 1em plus 0.5em minus 0.4em\relax
  Cambridge University Press, 2011.

\bibitem{yuan2022testing}
M.~Yuan, R.~Liu, Y.~Feng, and Z.~Shang, ``Testing community structure for
  hypergraphs,'' \emph{The Annals of Statistics}, vol.~50, no.~1, pp. 147--169,
  2022.

\bibitem{feldman2018complexity}
V.~Feldman, W.~Perkins, and S.~Vempala, ``On the complexity of random
  satisfiability problems with planted solutions,'' \emph{SIAM Journal on
  Computing}, vol.~47, no.~4, pp. 1294--1338, 2018.

\end{thebibliography}

\end{document}